\documentclass[journal,12pt,onecolumn,draftclsnofoot]{IEEEtran}

\usepackage[a4paper]{geometry}
\usepackage{cite}
\newcommand{\subparagraph}{}
\usepackage{titlesec}

\usepackage{graphicx}
\usepackage{hyperref}
\graphicspath{{./figures}}
\DeclareGraphicsExtensions{.pdf}
\usepackage[cmex10]{amsmath}
\usepackage{amssymb}
\usepackage{amsthm}
\usepackage{cleveref}
\usepackage{multirow}
\usepackage{bm}
\usepackage{algorithmic}
\usepackage{tikz}
\usetikzlibrary{shapes,arrows,positioning,plotmarks,shadows,calc,matrix,spy}
\usepackage{pgfplots}
\usepackage[disable]{todonotes}
\usepackage{balance}
\let\labelindent\relax
\usepackage{enumitem}
\usepackage[caption=false,font=footnotesize]{subfig}

\usepackage{fixltx2e}
\renewcommand{\IEEEQED}{\IEEEQEDopen}
\usepackage{pxfonts}

\newcommand{\figref}[1]{Figure~\ref{#1}}
\DeclareMathOperator{\snr}{snr}
\DeclareMathOperator{\snrae}{snr_{ae}}
\DeclareMathOperator{\snrbe}{snr_{be}}
\newcommand{\expv}{\mathbb{E}}

\DeclareMathOperator{\Var}{Var}
\newcommand{\hd}{\text{HD}}
\newcommand{\bxhd}{\left\lVert\bm b_x^{\hd}\right\rVert^2}
\newcommand{\exhd}{\left\lVert\bm e_x^{\hd}\right\rVert^2}
\newcommand{\fd}{\text{FD}}
\newcommand{\bxfd}{\left\lVert\bm b_x^{\fd}\right\rVert^2}
\newcommand{\exfd}{\left\lVert\bm e_x^{\fd}\right\rVert^2}

\newcommand{\rrhd}{R_{r}^{\hd}}
\newcommand{\rrhdup}{\overline{R}_{r}^{\hd}}
\newcommand{\rrfd}{R_{r}^{\fd}}
\newcommand{\rrfdlow}{\underline{R}_{r}^{\fd}}

\newcommand{\tf}{\mathbf}

\newcommand{\rskhd}{R_{sk}^{\hd}}
\newcommand{\rskhdup}{\overline{R}_{sk}^{\hd}}
\newcommand{\rskfd}{R_{sk}^{\fd}}
\newcommand{\rskfdlow}{\underline{R}_{sk}^{\fd}}
\titlespacing*{\subsubsection}{0pt}{1em}{0pt}

\newtheoremstyle{remarkmod}
  {\topsep}   
  {\topsep}   
  {\normalfont}  
  {0pt}       
  {\itshape} 
  {.}         
  {5pt plus 1pt minus 1pt} 
  {}          
\theoremstyle{remarkmod}

\newtheorem{proposition}{Proposition}
\newtheorem{definition}{Definition}

\newtheorem{remark}{\textit{Remark}}
\newtheorem{property}{Property}
\newtheorem{example}{Example}

\makeatletter
\newcommand*{\textoverline}[1]{$\overline{\hbox{#1}}\m@th$}
\makeatother

\tikzstyle{antenna} = [regular polygon,regular polygon sides=3, draw,shape
border rotate=180,minimum size=0.2pt,scale=0.4]

\ifCLASSOPTIONdraftcls
\else
\fi

\begin{document}
%
\title{Full-Duplex vs. Half-Duplex Secret-Key Generation\thanks{The result in this work was presented in part at the IEEE International Workshop on Information Forensics and Security, Rome, Italy, Nov. 2015 \cite{src:vogt2015full}. This work is supported by the Federal Ministry of Education and Research (BMBF) of the Federal Republic of Germany (F\"orderkennzeichen 16 KIS 0030, Prophylaxe).}}

\author{Hendrik Vogt, Zohaib~Hassan~Awan, Aydin Sezgin\\
Department of Electrical Engineering and Information Technology\\
Ruhr-Universit\"at Bochum, Germany\\
\{hendrik.vogt, zohaib.awan, aydin.sezgin\}@rub.de }

\maketitle

\IEEEpeerreviewmaketitle

\begin{abstract}
Full-duplex (FD) communication is regarded as a key technology in future 5G and Internet of Things (IoT) systems. In addition to high data rate constraints, the success of these systems depends on the ability to allow for confidentiality and security.
Secret-key agreement from reciprocal wireless channels can be regarded as a valuable supplement for security at the physical layer. In this work, we study the role of FD communication in conjunction with secret-key agreement. We first introduce two complementary key generation models for FD and half-duplex (HD) settings and compare the performance by introducing the key-reconciliation function.  Furthermore, we  study the impact of the so called probing-reconciliation trade-off, the role of a strong eavesdropper and analyze the system in the high SNR regime. We show that under certain conditions, the FD mode enforces a deteriorating impact on the capabilities of the eavesdropper and offers several advantages in terms of secret-key rate over the conventional HD setups. Our analysis reveals as an interesting insight that perfect self-interference cancellation is not necessary in order to obtain performance gains over the HD mode.
\end{abstract}

\section{Introduction}
The emerging deployment of devices with wireless connectivity in large
numbers --- commonly referred to as the Internet of Things (IoT) --- has
attracted significant attention in the research community. In IoT networks,
multiple nodes are allowed to interact with each other over a shared
medium. The communication between the nodes can be partitioned into two
types, namely, HD and FD modes. In HD mode, the transmitter and the receiver of each node are constrained to share their resources, i.e., time or frequency. More specifically, each node either listens to the transmission from the other nodes or broadcasts its own information over the network. As opposed to the HD mode, in the FD mode 
nodes can simultaneously transmit and receive
information on the same frequency band. 
From a practical viewpoint, a number of works
(see~\cite{src:khandani2013two,src:duarte2014design,src:bharadia2013full},
for instance), have recently proposed functional prototypes of FD systems. Due to the close
proximity of the transmitter and the receiver antennas, simultaneous
transmission and reception of information emanates a key issue of
 self-interference  (SI). The characterization and cancellation of SI
is the main challenge in the practical implementation of FD
systems~\cite{src:korpi2014full}.

Due to the openness of the wireless medium, the transmission between two
communicating parties can be overheard by other nodes in the network
for free. The leakage of information to unintended nodes in the network
may have serious consequences~\cite{src:jorswieck2015broadcasting}. In his seminal work, Wyner introduced a
basic wiretap channel model~\cite{src:wyner1975wire} to study secrecy  by taking
the attributes of the physical channel into account. A wiretap
channel consists of three nodes, two communicating users, i.e., Alice and
Bob and an external Eavesdropper (Eve) from whom the communication needs
to be protected. In~\cite{src:wyner1975wire}, Wyner characterized
the secrecy capacity of this model, when the Alice-to-Bob link is
stronger than the Alice-to-Eve link. The wiretap channel is extended to study a variety of multi-user settings, namely, the broadcast channel~\cite{src:awan2014achievable,src:chen2015on}, the multi-access channel~\cite{src:awan2013multiaccess}, the relay-eavesdropper channel~\cite{src:awan2012secure} and the wiretap channel with correlated sources~\cite{src:chen2014wiretap}. For a review of these and other related models, the interested reader may refer to~\cite{Liang} and references therein.

Alternatively, the confidential message can also be secured by the help of secret-key agreement~\cite{src:gamal2013Achievable}. In key agreement
systems, Alice encrypts the message with the help of a  secret key 
and broadcasts it over the network. It is assumed that the secret key is
available at only the legitimate nodes. The legitimate user Bob knows the secret key and can easily decode the confidential message. The
unintended node, Eve, is unaware of the secret key and hence is unable
to decode the confidential message. From information-theoretic viewpoint,
Ahlswede \emph{et al.}~\cite{src:ahlswede1993}, and
Maurer~\cite{src:maurer1993} study the problem of secret-key agreement
in bi-directional systems. They introduce the  source-type 
model, where all users observe information from jointly random sources. By means of public communication, Alice and Bob utilize these sources in order to harness an advantage over Eve~\cite{src:chou2014separation}. The maximum rate at which the secret keys can be generated  reliably, securely, and uniformly, is called as  secret-key capacity. 
The source-type model study in~\cite{src:ahlswede1993} has gained  much
attention in the research community and is extended to study a variety of
channels, namely,
rate-limited public communication~\cite{src:csiszar2000common,src:watanabe2010secret}, 
 and with Gaussian vector
sources~\cite{src:watanabe2011secret}. Recently, these works are
generalized to product
sources in~\cite{src:liu2014key}. In another
related work~\cite{src:nitinawarat2012secret}, the secret-key rate of
jointly
Gaussian sources under a quantization constraint is studied. A situation
where Alice can control the excitation of sources is considered
in~\cite{src:chou2015sender}. 

In practice, the main challenge in key agreement protocols is how to distribute the keys securely. In general, the channel states between two legitimate users,
i.e., Alice-to-Bob and Bob-to-Alice is known to be largely reciprocal,
while the channel from both legitimate nodes to Eve is not the same. Due to  the lack of similar observations, Eve is unable to
trace back the respective channel states. Thus, Alice and Bob can
utilize an advantage of common information to distill a secret key
which can be used to secure information~\cite{src:chen2014gaussian}. Some practical
implementations, for example in~\cite{src:mathur2008radio,src:liu2012exploiting,src:jana2013secret},
have demonstrated the feasibility of this approach. 
From an information-theoretic viewpoint, the authors of~\cite{src:wallace2009key,src:ye2010information} study the problem of key sharing by utilizing the leverage provided due to the reciprocity of wireless channels and establish different methods for key generation. Building on these works, in~\cite{src:khisti2012interactive} the fundamental trade-off between obtaining reciprocal CSI and randomness sharing is considered. In~\cite{src:zhou2014secret},  the authors  establish bounds on secret-key agreement in a two-way relay setting under active eavesdropping. As opposed to the previously mentioned results in which systems only operate in the HD mode to generate secret keys, in this work, we investigate this issue in the context of FD transmissions.

\begin{figure}
\centering 
\begin{tikzpicture}[
    block/.style={draw, drop shadow, fill=white, rectangle, minimum height=0.75cm, minimum width=3.5em},
	]
	\node[block] (alice) at (0,0) {Alice};	
	\node[block,below right=1cm and 1.75cm of alice] (eve) {Eve};
	\node[block,above right=1cm and 1.75cm of eve] (bob) {Bob};
	\node[antenna,above right=-0.0cm and 0.25cm of alice,anchor=west] (antalice) {};
	\node[antenna,above=0.25cm of eve] (anteve) {};
	\node[antenna,above left=-0.0cm and 0.25cm of bob,anchor=east] (antbob1) {};
	\node[below=0.15cm of antbob1] (antbob1b) {};
	
	\node[above right=-0.1cm and 0.6cm of alice,anchor=west] (chalice1) {};
	\node[below right=-0.4cm and 0.6cm of alice,anchor=west] (chalice2) {};
	\node[below right=-0.2cm and 0.6cm of alice,anchor=west] (chalice3) {};
	\node[above left=-0.1cm and 0.6cm of bob,anchor=east] (chbob1) {};
	\node[below left=-0.4cm and 0.6cm of bob,anchor=east] (chbob2) {};
	\node[below left=-0.2cm and 0.6cm of bob,anchor=east] (chbob3) {};
	
	\node[left=0.1cm of anteve] (cheve1) {};
	\node[right=0.1cm of anteve] (cheve2) {};
	\draw[thick] (antalice) |- (alice);
	\draw[thick] (eve) -- (anteve);
	\draw[thick] (antbob1) -- (antbob1b.center);
	\draw[thick] (antbob1b.center) -- (antbob1b -| bob.west);
	\draw[thick,->,dashed] (chalice1) -- (chbob1) node [above,pos=0.5] {$h_{ab}$};
	\draw[thick,->,dashed] (chbob2) -- (chalice2) node [below,pos=0.5] {$h_{ba}$};
	\draw[thick,->,dashed] (chalice3) -- (cheve1) node [below left,pos=0.6] {$h_{ae}$};
		\draw[thick,->,dashed] (chbob3) -- (cheve2) node [below right,pos=0.5] {$h_{be}$};
\end{tikzpicture}
\caption{The key agreement system model.} \label{fig:sys}
\end{figure}
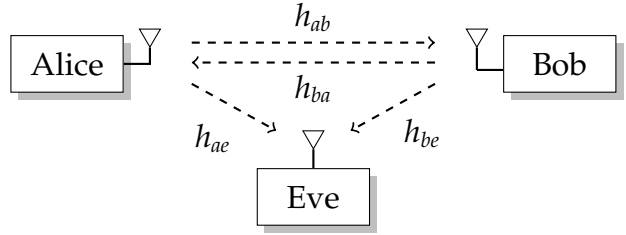

In this work, we study a model where three nodes Alice, Bob and Eve are equipped with a single antenna each as shown in Figure~\ref{fig:sys}, where Alice and Bob interact between each other to
establish a secret key that needs to be concealed from Eve. We assume that the Eve is a passive eavesdropper and only listens to the communication. The transmission protocol  consists of two phases. In the first phase,
Alice and Bob obtain common information by probing and estimating
the current state of their respective wireless channels. In the second
phase, Alice and Bob perform public
communication under a transmission rate constraint. The public communication by Alice is designed as one-shot key \emph{reconciliation}~\cite{src:chou2014separation}. Recall that, in the HD scenario, the communication nodes need to allocate distinct time or frequency resources for transmission or reception. In contrast to this, the extension to FD is far from straightforward. The main difficulty in FD mode arises due to occurrence of SI, since the nodes can transmit and receive information over the same frequency band. 

%
%

We now summarize the main contributions of this work. We first introduce the key-reconciliation function that serves as a metric to measure the performance of both HD and FD systems. Next, we show that there exists a fundamental trade-off between the channel probing and reconciliation phases. We also study a special case of a strong eavesdropper, in which the Eve is allowed to remove the contribution of receiver noise. We provide the region under which the FD mode outperforms the HD mode. Finally, we also investigate the asymptotic behavior of this model in high signal-to-noise ratio (SNR) regimes. For this setup, we establish certain conditions that shed light on the benefit of FD over HD mode. We illustrate our results with the help of some numerical examples. 

\vspace{.5em}
\noindent This paper is structured as follows. In section~\ref{sec:systemmodel},
the system models of both HD and FD modes are introduced. Section~\ref{sec:rates} defines the secret-key rate, more specifically, the
key-reconciliation function necessary for the evaluation of the system
performance. In section~\ref{sec:results}, the main contributions of this work are discussed. Finally, section~\ref{sec:conclusion} concludes
this paper by summarizing its contribution.

\vspace{.5em}
\textit{Notation:}
We  use following notations throughout this work. All vectors are
denoted by  small bold-face letter  $\tf{x}$, and  matrices by capital
bold-face letter  $\tf X$. The operator $\exp$ denotes the exponential function, $\ln$ denotes the logarithm with respect to base $e$, while $\log_2$ is used for base $2$. We denote the matrix transpose by
$(\cdot)^{T}$, expectation with respect to a random variable $x$ by
$\expv_x\left[ \cdot \right]$, entropy by $H(\cdot)$, differential entropy by $h(\cdot)$, Euclidean distance by $\lVert\cdot\rVert$ and ceiling operator by
$\lceil\cdot\rceil$. The operator $\overset{!}{=}$ allows equality of
left and right operand. The notation $\mathcal{N}\left(\bm 0, \tf
I_N\right)$ denote a Gaussian random variable with zero-mean and
identity covariance matrix of order $N\times N$.
\section{System model}
\label{sec:systemmodel}
We consider a key agreement model as shown in~\figref{fig:sys}. The system model consists of three nodes, i.e., two legitimate nodes, Alice and Bob, and an eavesdropper Eve, where each node is equipped with a single antenna. The parameter $h_{ab}$ is the channel which connects Alice-to-Bob and $h_{ba}$ is the channel which connects Bob-to-Alice. Furthermore, these two links may not necessarily fully correlate with each other. The parameter $h_{ae}$ is the channel which connects Alice-to-Eve and  $h_{be}$ is the channel which connects Bob-to-Eve. We assume that all channels comply with a real-valued flat-fading model.
We consider a block-fading environment where during one coherence block the channel coefficients remain constant and change to independent realizations in the next block. The key agreement protocol spans over $n$ such blocks, and is depicted in \figref{fig:seq}.
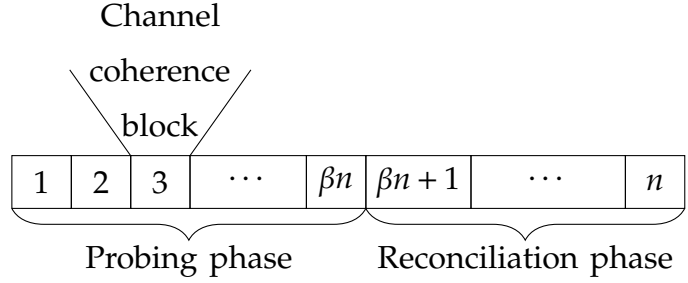
\begin{figure}
\centering 
\begin{tikzpicture}[
    phase/.style={draw, drop shadow, fill=white, rectangle, minimum height=2em, minimum width=5cm},
    block/.style={draw, fill=yellow, rectangle, minimum height=2em, minimum width=2em},
	field/.style={outer sep=0pt, draw, minimum height=7mm,
	 minimum width=#1cm,anchor=center,text width=#1cm, align=center}    
]
	\matrix (est) [matrix of nodes,column sep=-\pgflinewidth,
		 row 1/.style={nodes={field=.15}}] {
		 |[field=0.5] (one)|$1$  & |[field=0.5]|$2$ & |[field=0.5] (three)|$3$ & |[field=1.25]| $\cdots$ & |[field=0.5] (betan) |$\beta n$ &
		 |[field=1.1] (betanp)|$\beta n+1$ & |[field=1.75]| $\cdots$ & |[field=0.5] (last)|$n$ 
		 \\
		};		
		
	\node[above=0.1cm of three, text width= 2cm, align=center] (coh) {Channel coherence block};
	\node[above right=0cm and 0cm of coh.east] (cohright) {};
	\node[above left=0cm and 0cm of coh.west] (cohleft) {};
	
	\draw (cohright) -- (three.north east);
	\draw (cohleft) -- (three.north west);
		
	\draw[decorate,decoration={brace,amplitude=12pt,mirror},xshift=0pt,yshift=0cm]	(one.south west) -- (betan.south east) node [black,midway,yshift=-0.7cm] 
	{Probing phase};
	\draw[decorate,decoration={brace,amplitude=12pt,mirror},xshift=0pt,yshift=0cm]	(betanp.south west) -- (last.south east) node [black,midway,yshift=-0.7cm] 
	{Reconciliation phase};
\end{tikzpicture}
\caption{Key agreement protocol --- Channel estimation by probing and subsequent phase for key reconciliation by public communication.} \label{fig:seq}
\end{figure}
We assume that all interactions between Alice-to-Bob and vice versa is  authenticated  by some means. In what follows, we elucidate these two phases, i.e., probing and reconciliation phase, associated with HD and FD modes. 
 \vspace{-1 em}
\begin{remark}
In our system model, we consider an Eve with only one antenna. A more capable Eve can, however,  deploy multiple antennas in order to collect more observations. Subsequently, Eve can perform some information fusion to obtain a better estimate of the current channel state. For the sake of simplicity, we restrict ourselves to single antenna configuration, and emulate a stronger Eve by allowing  higher correlation of channel estimations to those of the legitimate nodes.
\end{remark}
\vspace{-2 em}
\subsection{Probing phase}
 \vspace{-.5 em}
In each block, Alice and Bob interact between each other to send pilot signals to estimate the current realization of the channel. These signals are also known to Eve. The probing phase consists of a total of $\beta n$ blocks, $0<\beta\leq 1$. In the $k$-th block, Alice gets one observation $x_k\in\mathbb{R}$, Bob gets $N_y$ observations in $\tf y_k\in\mathbb{R}^{N_y\times 1}$, and Eve gets $N_z$ observations in $\tf z_k\in\mathbb{R}^{N_z\times 1}$. The choice of parameters $N_y$  and $N_z$, is selected based on the operating mode (HD or FD) of the system. At the end of  $\beta n$  blocks, the independent and identically distributed (i.i.d.) Gaussian multiple vector source observations at Alice, Bob and Eve, are given by $x^{\beta n}$, $\tf y^{\beta n}$ and $\tf z^{\beta n}$, respectively. 

\vspace{-.75em}
\subsubsection{HD mode}
\vspace{.5em}
In HD mode, $N_y:=1$ and $N_z:=2$, i.e., $\tf y=y$ and $\tf z=\left(z_1,z_2\right)^T$. The channel observations in each block are given by
\begin{subequations}
\label{HD-model}
\begin{align}
\label{eq:aliceHdObs}
x &=\sqrt{\snr}\,h_{ba}+n_a, \\ 
\label{eq:bobHdObs}
y &=\sqrt{\snr}\,h_{ab}+n_b,  \\ 
\label{eq:eveHdObs1}
z_1 &=\sqrt{\snrae}\,h_{ae}+n_{ae},  \\
\label{eq:eveHdObs2}
z_2 &=\sqrt{\snrbe}\,h_{be}+n_{be}, 
\end{align}
\end{subequations}
where the channel coefficients $h_{ij}$, with $i\in \left\lbrace a,b \right\rbrace$ and $j\in \left\lbrace a,b,e \right\rbrace$, $i\neq j$,  are  jointly  Gaussian random variables with zero mean and unit variance. The correlation coefficients between different random variables are given by 
\begin{eqnarray}
&&\expv\left[ h_{ae}h_{be} \right]=\rho_{e}, \notag \\
&&\expv\left[h_{ae} h_{ba}\right]=\rho_{ae},\notag \\
&&\expv\left[ h_{be}h_{ba} \right]=\rho_{be},\notag \\
&&\expv\left[ h_{ba}h_{ab} \right]=\delta\rho_{ba}.
\end{eqnarray}
The parameter $0<\delta\leq 1$ captures the delay penalty  that consecutive measurements undergo in time-variant environments. This situation occurs when the reply of Bob to Alice's probing is already delayed further than the channel coherence time. The additive noise term $n_{ij}$ is i.i.d. with $n_{ij}\sim\mathcal{N}\left(0, 1\right)$. The signal-to-noise ratios (SNRs) at the legitimate nodes and the links Alice-to-Eve, Bob-to-Eve, are denoted by $\snr$, $\snrae$, and $\snrbe$, respectively. 
Note that the variances of both the channel and the noise variables are equal to one. In order to account for different scenarios, the parameters $\snr$, $\snrae$ and $\snrbe$ can be chosen appropriately. 
 
\begin{remark}
The legitimate nodes probe the channel \emph{once} per coherence block. Naturally, if the probing is repeated multiple times, the channel estimate is improved. However, in many communication standards, probing is primarily realized by a known preamble once sent before the actual payload data is transmitted. Thus, we assume that only \emph{one} channel estimate per coherence block is available. 
\end{remark}

\subsubsection{FD mode}
\vspace{.5em}
Next, we turn our attention to the FD mode. In this mode, $N_y:=1$, $N_z:=1$; and, thus $\tf y=y$ and $\tf z=z$. The channel observations in each block are given by
\begin{subequations}
\label{FD-model}
\begin{align}
\label{eq:aliceFdObs}
x &=\sqrt{\snr}\,h_{ba}+\alpha\sqrt{\snr}\,n_{\text{I}a}+n_a,\\ 
\label{eq:bobFdObs}
y &=\sqrt{\snr}\,h_{ab}+\alpha\sqrt{\snr}\,n_{\text{I}b}+n_b, \\ 
\label{eq:eveFdObs}
z &=\sqrt{\snrae}\,h_{ae}+\sqrt{\snrbe}\,h_{be}+n_e, 
\end{align}
\end{subequations}
where $n_{\text{I}a}$ and $n_{\text{I}b}$ denote the residual SI induced by simultaneous transmission and reception at Alice and Bob, respectively. This is due to the fact that, even after SI cancellation, some portion of the strong transmitter noise may still be present~\cite{src:korpi2014full}. The parameter $\alpha$ denotes the strength of the residual SI compared to the desired received signal, with $\alpha>0$, where $\alpha=1$ denotes the case when the residual SI is on the same power level as the desired signal. The statistics of the channel coefficients $h_{ij}$ hold verbatim as in the HD mode, except, that the penalty of consecutive measurements is no longer relevant since we have zero delay, i.e.,  $\expv\left[ h_{ba}h_{ab} \right]=\rho_{ba}$.
\begin{remark}
The channel observations in \eqref{eq:aliceFdObs} and \eqref{eq:bobFdObs} assume that the residual SI acts as independent noise. It is only linked to the transmitted signal by the power level, i.e., the second-order moment. This is an approximation, since the actual residual SI might still have some functional relationship to the transmitted signal. However, in this model, we assume that the SI cancellation is able to remove all parts of residual SI that depend on the transmitted signal directly.
\end{remark}

\subsection{Reconciliation phase} 
\vspace{.5em}
In this phase Alice and Bob use the remaining $(1-\beta)n$ blocks to convey public messages. Alice delivers information on the key reconciliation by the message $M_a\in\mathcal{M} = \{1,\hdots,2^{nR_{r}}\}$  to help Bob in key generation. We briefly discuss the role of the reconciliation, for details the reader may refer to~\cite{src:chou2014separation}. During the probing phase, Alice and Bob obtain $\beta n$ observations from channel estimation. These observations are correlated, but may not be equal and therefore require some  \textit{alignment}  which can be done through  public communication. Subsequently, at the beginning of the reconciliation phase, Alice quantizes its observation sequence and chooses a message for Bob as a function of the quantized sequence in the spirit of Wyner-Ziv coding with side information at the receiver~\cite{src:el2011network}. Since the channel demands a rate constraint, the quality of quantization allows to control the required reconciliation rate. 

We now elaborate on the role of  Bob. It is evident that FD capability is meaningful only in the context of two-way communication. However, as mentioned before,  the reconciliation of Alice-to-Bob is a one-way transmission. In practical situations, the transmissions should be designed in a two-way fashion to fully utilize the leverage offered by FD transmission. Therefore, we assume that not only Bob is receiving the key reconciliation message from Alice, it is allowed to broadcasts a public message to another external node. We provide a simple example to highlight the impact of this communication.
\begin{example}
Consider a setup for the 3-user key agreement system as shown in Figure~\ref{fig:example}. 
\begin{figure}
\centering 
\begin{tikzpicture}[
    block/.style={draw, drop shadow, fill=white, rectangle, minimum
height=0.75cm, minimum width=3.5em},
	]
	\node[block] (alice) at (0,0) {Alice};	
	\node[block,right=2cm of alice] (bob) {Bob};
	\node[block,below right=1cm and 1cm of bob] (eve) {Eve};
	\node[block,right=2cm of bob] (charlie) {Charlie};
	\node[block,right=2cm of charlie] (next) {};
	
	\draw[thick,->] (alice) to node[above] {$M_a$} (bob);
	\draw[thick,->] (bob) to node[above] {$M_b$} (charlie);
	\draw[thick,->,dotted] (charlie) to (next);
	\draw[thick,->,dotted] (next) to ($(next) + (2,0)$);
	\draw[thick,->,dashed] (alice) |-  (eve);
	\draw[thick,->,dashed] (bob) to  (eve);
	\draw[thick,->,dotted] (charlie) to  (eve);
	\draw[thick,->,dotted] (next) |-  (eve);
	\draw[very thick,dotted] ($(charlie)!0.5!(next) + (0,0.5)$) to
($(charlie)!0.5!(next) + (0,-2)$);
\end{tikzpicture}
\caption{Example of reconciliation for 3-user key agreement.}
\label{fig:example}
\end{figure}
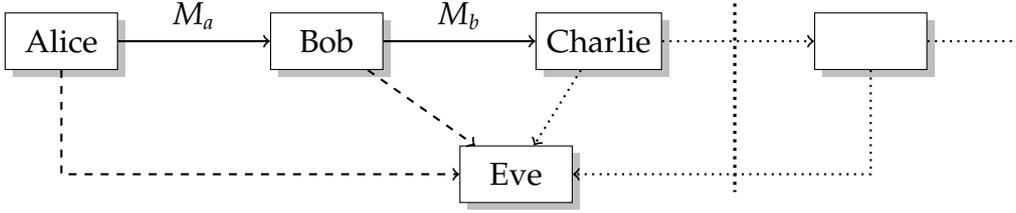
This model consists of four nodes, three legitimate users Alice, Bob, Charlie and the eavesdropper Eve.  The key agreement protocol is depicted in Figure~\ref{fig:threeprotocol} for the FD mode. 
\begin{figure}
\centering 
\def\blockheight{1.75}
\def\blockwidth{8.5}
\begin{tikzpicture}[
    block/.style={draw, drop shadow, fill=white, rectangle, minimum
height=\blockheight cm, minimum width=\blockwidth cm},
    subblock/.style={draw, fill=white, rectangle, minimum
height=\blockheight*0.5 cm, minimum width=\blockwidth*0.5 cm},
    subblock2/.style={draw, fill=white, rectangle, minimum
height=\blockheight*0.5 cm, minimum width=\blockwidth*0.75 cm},
	]
	\matrix(struct)[matrix of nodes,ampersand replacement=\&,
	    row 1/.style={nodes={text width=4cm,align=center,minimum
height=1cm, minimum width=2cm}}
	] at (0,0) {
	   Probing Phase \& Reconciliation Phase \\
	   \node[rectangle, minimum
	   height=\blockheight cm, minimum width=\blockwidth cm] (prob) {}; \&
\node[block, minimum
width=\blockwidth*0.75 cm] (rec){}; \\
	};
	
	\node[subblock,drop shadow] at
($(prob.center)+(-\blockwidth*0.25,\blockheight*0.25)$) {Alice
$\Longleftrightarrow$ Bob};
	\node[subblock,drop shadow] at
($(prob.center)+(\blockwidth*0.25,-\blockheight*0.25)$) {Bob
$\Longleftrightarrow$ Charlie};
	\node[subblock2] at ($(rec.center)+(0,\blockheight*0.25)$) {Alice
$\Longrightarrow$ \textcolor{blue}{Bob}  };
	\node[subblock2] at ($(rec.center)+(0,-\blockheight*0.25)$)
{\textcolor{blue}{Bob} $\Longrightarrow$ Charlie};
\end{tikzpicture}
\caption{A sketch of the transmission protocol for the 3-user key
agreement system.} \label{fig:threeprotocol}
\end{figure}
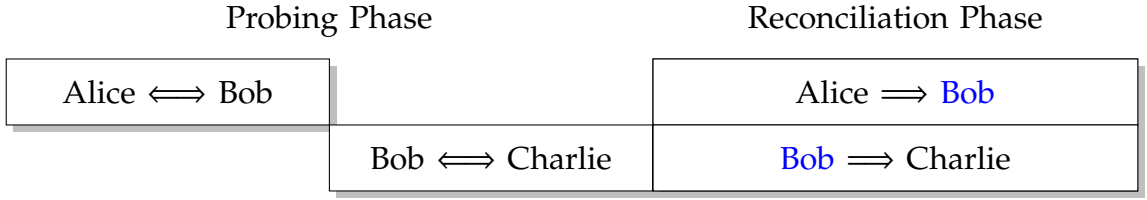
We assume that during the probing phase, the user pairs Alice-Bob and Bob-Charlie probe the channel in different coherence blocks. Therefore, the observation sequences of Alice-Bob and Bob-Charlie pairs are independent. In the reconciliation phase, Alice transmits the message $M_a$ to Bob, and, Bob sends $M_b\in\mathcal{M} = \{1,\hdots,2^{nR_{r}}\}$ to Charlie for key reconciliation. As a consequence, Bob is both transmitting and receiving information at the same time in FD mode, or, in HD mode, by time-sharing. This argument can be extended to the arbitrary large independent key agreement networks.  \flushright \IEEEQED
\end{example}

\noindent For the sake of simplicity, we restrict ourselves to a single setup of key agreement. The ingredients that remain crucial are that Bob is actively using the transmitter while receiving the reconciliation message from Alice.
Thus, the Alice-to-Bob link achieves the reconciliation rate $R_{r}$ as
\begin{align}
R_{r} &\geq \frac{1}{n}H(M_a).
\end{align}
We assume that Bob has only imperfect CSI available. Thus, the communication for reconciliation cannot achieve the capacity of a fast-fading channel with perfect CSI~\cite{src:tse2005fundamentals}, but rather a lower bound on it.


\subsubsection{HD mode}
\vspace{.5em}
At Bob, the channel input-output relationship is given by
\begin{align}
y_{r}=h_{ab}x_r+n_r,
\label{eq:commHd}
\end{align}
where $x_r\sim\mathcal{N}\left(0, \snr\right)$ is independent of $h_{ab}$,
and $n_r\sim\mathcal{N}\left(0, 1\right)$.
We define $\hat{h}_{ab}$ as the part of the channel state that the receiver estimates
by minimum-mean squared error (MMSE) estimation. The variance of the estimation error is given by
\begin{align}
\label{sigma-hd}
\sigma^2_{\text{HD}}=1-\frac{\expv \left[ h_{ab} y_r\right]^2}{\expv
\left[ y_r^2 \right]}=\frac{1}{1+\snr}.
\end{align} 

\subsubsection{FD mode}
Recall that, since Bob is transmitting and receiving at the same time, in FD mode the communication is subjected to self-interference, and thus Bob receives
\begin{align}
y_{r}=h_{ab}x_r+\alpha\sqrt{\snr}\; n_{\text{I}r} + n_r,
\label{eq:commFd}
\end{align}
where  $x_r\sim\mathcal{N}\left(0, \snr\right)$, $n_{\text{I}r}\sim\mathcal{N}\left(0, 1\right)$ is the residual noise after interference cancellation and $n_r\sim\mathcal{N}\left(0, 1\right)$ is independent additive white Gaussian noise. The receiver obtains imperfect channel knowledge by MMSE from probing in the same manner as in the HD case. The variance of the estimation error yields 
\begin{align}
\label{sigma-fd}
\sigma^2_{\text{FD}}=\frac{1+\alpha^2\snr}{1+\snr(1+\alpha^2)}.
\end{align}

\subsection{Key generation}
Following the probing and reconciliation phase, Alice computes a secret key $K_a\in\mathcal{K}$ from the observations $x^{\beta n}$,  where $\mathcal{K}=\left\lbrace 1,\ldots,2^{nR_{sk}} \right\rbrace$ denotes a finite alphabet set and $R_{sk}$ is the secret-key rate. From the public message $M_a$ and the observations $\bm y^{\beta n}$, Bob obtains a key $K_b\in\mathcal{K}$, while Eve tries to reconstruct the key from observations $\bm z^{\beta n}$ and the public message $M_a$. The performance is measured by the probability of error
\begin{align}
\label{eq:probError}
\lim\limits_{n\to\infty}\Pr\left\lbrace K_a \neq K_b \right\rbrace =0,
\end{align}
and both strong uniformity and secrecy~\cite[p. 116]{src:bloch2011physical}
\begin{align}
\label{eq:uniSecrecy}
\lim\limits_{n\to\infty}\left[\log_2\lceil 2^{nR_{sk}}\rceil - H\left(K_a|M_a,\bm z^{\beta n}\right)\right]=0.
\end{align}
For arbitrarily large $n$, the secret-key rate $R_{sk}$ is achievable, if \eqref{eq:probError} and \eqref{eq:uniSecrecy} are satisfied.
\begin{definition}
The secret-key rate~$R_{sk}$ with respect to a certain reconciliation rate~$R_r$, is denoted by the key-reconciliation function~$R_{sk}\left(R_r\right)$.
\end{definition}

\section{Computation of Secret Key Rates} 
\label{sec:rates}
In this section, we derive the secret-key rates that serve as a metric for performance evaluation. 
\subsection{Reconciliation rates}
First, we provide the description to compute the reconciliation rate expression for both HD and FD modes. We consider a public transmission from Alice to Bob during the reconciliation phase. This communication utilizes the leverage offered by FD capability, since the link Alice-to-Bob obtains a rate improvement by a factor of one-half compared to HD mode~\cite{src:ahmed2013rate}. 
 \vspace{-1em}
\subsubsection{HD mode}
\vspace{-.25em}
 For the HD mode in~\eqref{eq:commHd}, the reconciliation rate per block between Alice-to-Bob channel is given by 
\begin{align}
\label{eq:rphdintro}
\rrhd = \frac{\left(1-\beta\right)}{2} I(x_r;y_r|\hat{h}_{ab}),
\end{align}
where the factor of one-half represents the resource share in HD mode. 

\noindent In what follows, we provide an upper bound on the reconciliation rate of~\eqref{eq:rphdintro}. We can write
\begin{align}
I(x_r;y_r|\hat{h}_{ab})
&\overset{(a)}{\leq} I(x_r;y_r|h_{ab})\nonumber \\
\label{eq:muiHD}
&=\frac{1}{2}\expv_{h_{ab}}\left[\log_2\left(1+h_{ab}^2\snr\right)\right],
\end{align}
where $(a)$ follows since conditioning reduces entropy. Subsequently, by plugging~\eqref{eq:muiHD} into~\eqref{eq:rphdintro}, we get
\begin{align}
\rrhd & \leq \frac{1-\beta}{4}\expv_{h_{ab}}\left[\log_2\left(1+ h_{ab}^2\snr \right)\right]\nonumber\\
& \overset{(b)}{\leq}\frac{1-\beta}{4}\log_2\left(1+ \expv_{h_{ab}}\left[h_{ab}^2\right]\snr \right)\nonumber\\
\label{eq:rpHdBound}
&= \frac{1-\beta}{4}\log_2\left(1+ \snr \right)=:\rrhdup,
\end{align}
where $(b)$ follows from Jensen's inequality~\cite[p. 27]{src:cover2012elements}.

\subsubsection{FD mode}
In contrast to HD mode, with FD capabilities Bob does not need to split frequency or time resources. Thus, the  Alice-to-Bob link is able to support twice the rate compared to the HD mode. The supported reconciliation rate per block between Alice-to-Bob channel is given by
\begin{align}
\label{eq:rpfdintro}
&\rrfd= 
\left(1-\beta\right) I(x_r;y_r|\hat{h}_{ab}).
\end{align}
  From~\eqref{eq:rpfdintro}, we get
\begin{align}
\label{point-to-point}
I(x_r;y_r|\hat{h}_{ab})&=h(x_r)-h(x_r|y_r,\hat{h}_{ab}).
\end{align}
The right-hand side term in~\eqref{point-to-point} can be bounded as 
\begin{align}
h(x_r|y_r,\hat{h}_{ab}) &\overset{(c)}{=} h(x_r-cy_r|y_r,\hat{h}_{ab}) \nonumber\\
&\overset{(d)}{\leq} h(x_r-cy_r|\hat{h}_{ab}) \nonumber\\
\label{eq:entropyHDbound}
&\overset{(e)}{\leq} \frac{1}{2}\expv_{\hat{h}_{ab}}\left[\log_2\left( (2\pi e) \Var\left[ x_r-cy_r |\hat{h}_{ab}\right] \right)\right], 
\end{align} 
where $(c)$ follows since adding a constant with $c>0$ does not change entropy, $(d)$ follows due to the fact that conditioning reduces entropy; and, $(e)$ follows because Gaussian distribution maximizes conditional entropy for a fixed variance. The quantity $\Var\left[ x_r-cy_r |\hat{h}_{ab}\right]$ is given by
\begin{align}
\label{eq:varhd}
\Var\left[ x_r-cy_r  |\hat{h}_{ab}\right]=
\frac{\snr\left(1+\snr(\sigma_{\text{FD}}^2+\alpha^2)\right)}{1+\snr(\sigma_{\text{FD}}^2+\alpha^2+\hat{h}_{ab}^2)},
\end{align}
where $c$ is chosen by minimizing the  mean-square error (MSE) when estimating $x_r$ from $y_r$, and is
\begin{align}
c&=\frac{\hat{h}_{ab} \snr}{1+\snr(\sigma^2_{\text{FD}}+\alpha^2+\hat{h}_{ab}^2)}.
\end{align}
Then, by plugging~\eqref{eq:varhd} in~\eqref{eq:entropyHDbound}, and replacing it in~\eqref{point-to-point}, we get
\begin{equation}
\label{eq:muiFD}
\hspace{-1em}I(x_r;y_r|\hat{h}_{ab})\geq  \hspace{.5em}\frac{1}{2}\expv_{\hat{h}_{ab}}\left[\log_2\left( 1+\hat{h}_{ab}^2\frac{\snr(1+\snr(1+\alpha^2))}{1+2\snr+\alpha\snr(1+\snr)}\right)\right].
\end{equation}
Finally, the reconciliation rate for FD is given by
\begin{align}
\label{eq:rpfdInter}
\rrfd &\geq \frac{1-\beta}{2}\expv_{\hat{h}_{ab}}\left[ \log_2\left(1+c_{\text{FD}}\hat{h}_{ab}^2\right) \right]
\end{align}
where for convenience, we set 
\begin{align}
\label{eq:cfd}
c_{\text{FD}}&:=  \frac{\snr(1+\snr(1+\alpha^2))}{1+2\snr+\alpha\snr(1+\snr)}.
\end{align} 
We can further bound~\eqref{eq:rpfdInter} as follows.
\begin{align}
\rrfd&\geq\frac{1-\beta}{2}\expv_{\hat{h}_{ab}}\left[ \log_2\left(1+c_{\text{FD}}\exp\left(\ln \hat{h}_{ab}^2\right)\right) \right] \nonumber\\
\label{eq:rpFdBoundIntermediate}
& \overset{(f)}{\geq} \frac{1-\beta}{2} \log_2\left(1+c_{\text{FD}}\exp\left(\expv_{\hat{h}_{ab}}\left[\ln \hat{h}_{ab}^2\right]\right)\right), 
\end{align}
\noindent where $(f)$ follows due to the fact that $f(x)=\log_2\left(1+ce^{x}\right)$ is a convex function~\cite{src:oyman2002tight}. The expectation over $\expv_{\hat{h}_{ab}}\left[\ln \hat{h}_{ab}^2 \right]$  yields
\begin{align}
\expv_{\hat{h}_{ab}}\left[\ln \hat{h}_{ab}^2 \right] &= \frac{1}{\sqrt{2\pi\sigma^2_{\hat{h}_{ab}}}}\int_{\mathbb{R}}\ln x^2 \exp\left(-\frac{x^2}{2\sigma^2_{\hat{h}_{ab}}}\right) \text{d}x \nonumber\\
&\overset{(g)}{=} \ln\sigma^2_{\hat{h}_{ab}} + \frac{1}{\sqrt{2\pi}}\int_{\mathbb{R}}\ln y^2 \exp\left(-\frac{y^2}{2}\right) \text{d}y \nonumber\\
&\overset{(h)}{=} \ln\sigma^2_{\hat{h}_{ab}} + \Psi\left(0.5\right)+\ln 2 \nonumber\\
\label{eq:expClosed}
&\overset{(i)}{=} \ln\frac{\sigma^2_{\hat{h}_{ab}}}{2} -\gamma. 
\end{align}
where $(g)$ holds due to the substitution $x=\sigma_{\hat{h}}y$, $(h)$ follows because the integral can be solved in closed form~\cite{src:lee2012bayesian} and is given by the Digamma function $\Psi(x)$, and $(i)$ follows from the Euler-Mascheroni constant $\gamma\approx 0.57721$. Replacing~\eqref{eq:expClosed} in \eqref{eq:rpFdBoundIntermediate}, we get
\begin{align}
\rrfd &\geq \frac{1-\beta}{2} \log_2\left(1+c_{\text{FD}}\frac{\sigma^2_{\hat{h}_{ab}}}{2}e^{-\gamma}\right) \nonumber\\
&\overset{(j)}{=} \frac{1-\beta}{2} \log_2\left(1+\frac{1}{2}e^{-\gamma}\frac{\snr^2}{1+2\snr+\alpha^2\snr(1+\snr)}\right) \nonumber \\
\label{eq:rpFdBound}
&=:\rrfdlow,
%
\end{align}
where $(j)$  follows since
\begin{align*}
\sigma^2_{\hat{h}_{ab}}=1-\sigma^2_{\text{FD}}=\frac{\snr}{1+\alpha^2\snr+\snr}.
\end{align*}

\subsection{Secret-key rates}
We now provide the secret-key rate for the Gaussian sources with rate-limited public communication for key reconciliation.
\begin{proposition}\label{secret-key}
Let Alice, Bob and Eve observe a zero-mean Gaussian multiple vector source $\left(x,\tilde{\bm y},\tilde{\bm z}\right)$ 
\begin{subequations}
\label{model-ref}
\begin{align}
\label{eq:normSrcY}
\tilde{\bm y} &= \bm b x + \bm w_y, \\
\label{eq:normSrcZ}
\tilde{\bm z} &= \bm e x + \bm w_z,
\end{align}
\end{subequations}
where $\bm b\in \mathbb{R}^{N_y}$, $\bm e\in \mathbb{R}^{N_z}$, $x\sim\mathcal{N}\left(0, \sigma^2_x\right)$, $\bm w_y\sim\mathcal{N}\left(\bm 0, \bm I_{N_y}\right)$ and $\bm w_z\sim\mathcal{N}\left(\bm 0, \bm I_{N_z}\right)$, and recall that $\beta$ denotes the ratio of coherence blocks available for collecting reciprocal channel states at each node. Then, the rate region $\left(R_{sk}, R_r\right)$ is the union of all achievable rate pairs satisfying
\begin{subequations}
\label{CR-KR}
\begin{align}
\label{eq:commRate}
R_r & \geq \frac{\beta}{2}\log_2 \frac{\sigma^2_x}{\sigma^2_{x|u}}-\frac{\beta}{2}\log_2 \left|\frac{\bm I_{N_y}+\bm b \bm b^T\sigma^2_x}{\bm I_{N_y}+ \bm b\bm b^T\sigma^2_{x|u}}\right|, \\
\label{eq:keyRate}
R_{sk} & \leq \frac{\beta}{2}\log_2 \left|\frac{\bm I_{N_y}+\bm b \bm b^T\sigma^2_x}{\bm I_{N_y}+\bm b \bm b^T\sigma^2_{x|u}}\right|-\frac{\beta}{2}\log_2 \left|\frac{\bm I_{N_z}+\bm e \bm e^T\sigma^2_x}{\bm I_{N_z}+\bm e \bm e^T\sigma^2_{x|u}}\right|,
\end{align} 
\end{subequations}
for some $\sigma^2_x\geq\sigma^2_{x|u}>0$, where $\sigma^2_{x|u}$ denotes the conditional covariance of $x$ given a (Gaussian) random variable $u$.
\end{proposition}
\begin{proof}
The proof of~\eqref{CR-KR} follows along similar lines as in~\cite[Section V]{src:watanabe2011secret} and is omitted for brevity. 
\end{proof}
\vspace{.5em}
\noindent Next, we provide a useful proposition that is used to establish the results in this work.
\begin{proposition} 
\label{sec:prop1}
There exist non-singular matrices $\bm A_y\in\mathbb{R}^{N_y\times N_y}$, $\bm A_z\in\mathbb{R}^{N_z\times N_z}$, for which the achievable rate region $\left(R_r,R_{sk}\right)$ defined in \eqref{eq:commRate} and \eqref{eq:keyRate}, also holds for~\eqref{HD-model} and~\eqref{FD-model}, respectively. 
\end{proposition}
\begin{IEEEproof}
The computation of the achievable rate region in~\eqref{CR-KR} for the model~\eqref{model-ref} depends on joint probability distributions only through marginals $\left(x,\bm y\right)$ and $\left(x,\bm z\right)$~\cite[Appendix C]{src:watanabe2010secret}. Let $\bar{\bm y}$ and $\bar{\bm z}$ be the Gaussian random variables with the same second-order moments as $\bm y$ and $\bm z$, respectively, and the same joint statistics with $x$. In what follows, we elaborate the connection between $\bar{\bm z}$  and $\bm z$. Similar arguments can be used to show the relationship between $\bar{\bm y}$  and $\bm y$ and is omitted. Let 
 \begin{align}
 \label{eq:rvZbar}
 \bar{\bm z} &:= \bm A_z\tilde{\bm z}=\bm A_z\bm e x + \bm A_z\bm w_z,
 \end{align}
 where  $\bm A_z\in\mathbb{R}^{N_z\times N_z}$ is non-singular. Since $\bm A_z$ is invertible, the transformation $\bm A_z\tilde{\bm z}$ provides a sufficient statistic and is information lossless. 
Subsequently, we need the same joint and second-order marginal statistics of $\bm z$ and $\bar{\bm z}$, i.e., 
\begin{align}
\bm\sigma_{\bm z x} &= \expv\left[ \bm z x \right] \notag\\
\label{eq:aze}
&\overset{!}{=} \expv\left[ \bar{\bm z} x \right] =\bm A_z\bm e\sigma_x^2,
\end{align}
\vspace{-0.5em}
and
\begin{align}
\bm\Sigma_{\bm z} &= \expv\left[ \bm z \bm z^T \right] \notag\\
&\overset{!}{=} \expv\left[ \bar{\bm z} \bar{\bm z}^T \right] \notag\\
&= \bm A_z\bm e\left(\sigma_x^2\bm A_z\bm e\right)^T+\bm A_z\bm A_z^T\notag\\
\label{eq:covarz}
&=\sigma^{-2}_x\bm \sigma_{\bm zx}\bm \sigma_{\bm zx}^T+\bm A_z\bm A_z^T.
\end{align}
From \eqref{eq:covarz}, we get
\begin{align}
\label{eq:azaz}
\bm A_z \bm A_z^T = \bm \Sigma_{\bm z}-\bm \sigma_{\bm zx}\sigma^{-2}_x\bm \sigma_{\bm zx}^T=:\bm \Sigma_{\bm z|x}.
\end{align}
\end{IEEEproof} 
\begin{remark}
By construction, $\bm \Sigma_{\bm z}$, $\bm \Sigma_{\bm z|x}$ have to be positive semi-definite matrices. However, in order to satisfy the conditions of Proposition~\ref{sec:prop1}, we also need to exclude all the cases where $\bm \Sigma_{\bm z}$, $\bm \Sigma_{\bm z|x}$ have at least one eigenvalue of zero. 
\end{remark}
The next step is to compute the parameter vectors $\bm b$ and $\bm e$. We now provide the computation of $\lVert\bm e\rVert^2$. The computation of $\lVert\bm b\rVert^2$ follows analogously and is omitted. Starting from \eqref{eq:aze}, the squared norm of the $\bm e$ is given by
\ifCLASSOPTIONdraftcls
\begin{align}
\lVert\bm e\rVert^2 &=\sigma^{-4}_x\bm\sigma_{\bm z x}^T\left(\bm A_z\bm A_z^{T}\right)^{-1}\bm\sigma_{\bm z x} \nonumber \\
&\overset{(a)}{=} \sigma^{-4}_x\bm\sigma_{\bm zx}^T\left(\bm \Sigma_{\bm z}-\bm \sigma_{\bm zx}\sigma^{-2}_x\bm \sigma_{\bm zx}^T\right)^{-1}\bm\sigma_{\bm zx} \nonumber \\
&\overset{(b)}{=} \sigma^{-4}_x\bm\sigma_{\bm zx}^T
\left(
\bm\Sigma_{\bm z}^{-1}\right.\left.+\bm\Sigma_{\bm z}^{-1}\bm\sigma_{\bm zx}
\left(\sigma_{x}^{2}-\bm\sigma_{\bm zx}^T\bm\Sigma_{\bm z}^{-1}\bm\sigma_{\bm zx}\right)^{-1}
\bm\sigma_{\bm zx}^T\bm\Sigma_{\bm z}^{-1}
\right)
\bm\sigma_{\bm zx}\nonumber\\
\label{eq:parametere}
&= \sigma^{-2}_x\frac{\bm\sigma_{\bm zx}^T\bm\Sigma_{\bm z}^{-1}\bm\sigma_{\bm zx}
}{\sigma^2_x-\bm\sigma_{\bm zx}^T\bm\Sigma_{\bm z}^{-1}\bm\sigma_{\bm zx}},
\end{align}
\else
\begin{align}
\lVert\bm e\rVert^2 &=\sigma^{-4}_x\bm\sigma_{\bm z x}^T\left(\bm A_z\bm A_z^{T}\right)^{-1}\bm\sigma_{\bm z x} \nonumber \\
&\overset{(a)}{=} \sigma^{-4}_x\bm\sigma_{\bm zx}^T\left(\bm \Sigma_{\bm z}-\bm \sigma_{\bm zx}\sigma^{-2}_x\bm \sigma_{\bm zx}^T\right)^{-1}\bm\sigma_{\bm zx} \nonumber \\
&\overset{(b)}{=} \sigma^{-4}_x\bm\sigma_{\bm zx}^T
\left(
\bm\Sigma_{\bm z}^{-1}\right. \nonumber\\
&\phantom{{}=}\left.+\bm\Sigma_{\bm z}^{-1}\bm\sigma_{\bm zx}
\left(\sigma_{x}^{2}-\bm\sigma_{\bm zx}^T\bm\Sigma_{\bm z}^{-1}\bm\sigma_{\bm zx}\right)^{-1}
\bm\sigma_{\bm zx}^T\bm\Sigma_{\bm z}^{-1}
\right)
\bm\sigma_{\bm zx}\nonumber\\
\label{eq:parametere}
&= \sigma^{-2}_x\frac{\bm\sigma_{\bm zx}^T\bm\Sigma_{\bm z}^{-1}\bm\sigma_{\bm zx}
}{\sigma^2_x-\bm\sigma_{\bm zx}^T\bm\Sigma_{\bm z}^{-1}\bm\sigma_{\bm zx}},
\end{align}
\fi

\noindent where $(a)$ follows from \eqref{eq:azaz} and $(b)$ is due to the Woodbury identity~\cite[eq. (156)]{src:petersen2012matrix}.  

 Following steps similar to as shown above, $\lVert\bm b\rVert^2$ is given by
\begin{align}
\label{eq:parameterb}
\lVert\bm b\rVert^2= \sigma^{-2}_x\frac{\bm\sigma_{\bm yx}^T\bm\Sigma_{\bm y}^{-1}\bm\sigma_{\bm yx}
}{\sigma^2_x-\bm\sigma_{\bm yx}^T\bm\Sigma_{\bm y}^{-1}\bm\sigma_{\bm yx}}.
\end{align}

\subsection{Key-reconciliation function}
The following proposition provides the secret-key rate as a function of reconciliation rate. 
\begin{proposition}
\label{prop1}
The key-reconciliation function is given by
\ifCLASSOPTIONdraftcls
\begin{align}
\label{eq:keyfunc}
R_{sk}(R_r) = 
 \frac{\beta}{2}\log_2\frac{1-4^{-\frac{R_r}{\beta}}(\lVert\bm b_x\rVert^2-\lVert\bm e_x\rVert^2)+\lVert\bm b_x\rVert^2}{1+\lVert\bm e_x\rVert^2}.
\end{align}
\else
\begin{align}
&\hspace{-.5em} R_{sk}(R_r) = \nonumber\\
\label{eq:keyfunc}
 &\frac{\beta}{2}\log_2\frac{1-2^{-2\frac{R_r}{\beta}}(\lVert\bm b_x\rVert^2-\lVert\bm e_x\rVert^2)+\lVert\bm b_x\rVert^2}{1+\lVert\bm e_x\rVert^2}.
\end{align}
\fi
\end{proposition}
\begin{IEEEproof}
We begin the proof by first considering $R_r$ --- the minimum possible reconciliation rate in~\eqref{eq:commRate} and define
\begin{subequations}
\begin{align}
\label{pb}
\lVert\bm b_x\rVert^2 &=\sigma^2_x\lVert\bm b\rVert^2,\\
\label{pe}
\lVert\bm e_x\rVert^2 &=\sigma^2_x\lVert\bm e\rVert^2.
 \end{align}
\end{subequations}
From the determinant 
identity~\cite[eq. (24)]{src:petersen2012matrix}, we know that 
\begin{align}
\left|\bm I_{N_y}+\bm b \bm b^T\sigma^2_{x|u}\right|=1+\sigma^2_{x|u}\lVert\bm b\rVert^2.
\end{align}
Subsequently, from~\eqref{eq:commRate}, we get 
\begin{align}
\label{eq:aux}
\sigma^2_{x|u}=\frac{\sigma^2_x}{4^{R_r/\beta}\left(1+\lVert\bm b_x\rVert^2\right)- \lVert\bm b_x\rVert^2}.
\end{align}
Finally, by plugging~\eqref{eq:aux} into \eqref{eq:keyRate} and repetitively using the determinant identity~\cite[eq. (24)]{src:petersen2012matrix},
we get~\eqref{eq:keyfunc}.
\end{IEEEproof}
 \vspace{.5em}
\noindent From Proposition~\ref{prop1}, we  deduce the following property.
\begin{property}
\label{sec:propskpos}
The key-reconciliation function $R_{sk}(R_r)$ is positive if and only if 
\begin{enumerate}[label=(\roman*)]
\item $R_r>0$,\label{property-R}
\item $\lVert\bm b_x\rVert^2>\lVert\bm e_x\rVert^2$.
\label{property-sec}
\end{enumerate}  
\end{property} 
\begin{IEEEproof}
Let $R_{sk}(R_r)>0$, i.e., the numerator term inside $\log$ function in~\eqref{eq:keyfunc}  must be larger than the denominator. This implies that 
\begin{equation*}
4^{-R_r/\beta}(\lVert\bm b_x\rVert^2-\lVert\bm e_x\rVert^2)<\lVert\bm b_x\rVert^2-\lVert\bm e_x\rVert^2.
\end{equation*}
The exponential term is always smaller or equal to one due to the non-negativity of reconciliation rate, therefore the inequality can only be fulfilled if $\lVert\bm b_x\rVert^2>\lVert\bm e_x\rVert^2$ and $R_r>0$. Conversely, assuming the conditions (i) and (ii) are satisfied, then one can strictly lower bound \eqref{eq:keyfunc} by removing the exponential term, which in turn yields $R_{sk}(R_r)>0$.  
\end{IEEEproof}
\vspace{.5em}
\noindent Next, we establish the appropriate representations for the key-reconciliation function for both HD and FD modes. 
\subsubsection{HD mode}
\label{HD-Key-compuation}
For HD mode~\eqref{HD-model}, we denote the key-reconciliation function as $\rskhd$. The computation of $\rskhd$ follows along similar lines as in~\eqref{eq:keyfunc} with the specific choice of parameters 
 $\exhd$, and $\bxhd$. This leads to
\begin{align}
\label{eq:keyfuncHDup}
\rskhd\left(\rrhd\right) \overset{(a)}{\leq} \rskhd\left(\rrhdup\right) =: \rskhdup,
\end{align}
where $(a)$ follows due to~\eqref{eq:rpHdBound}. 
Details on the computation of the parameters $\exhd$, and $\bxhd$, are relegated to Appendix~\ref{sec:appGenHD}.
%

\subsubsection{FD mode}
\label{FD-Key-compuation}
In FD mode~\eqref{FD-model}, the key-reconciliation function is denoted by $\rskfd$.
The computation of $\rskfd$ follows along similar lines as in~\eqref{eq:keyfunc} with the specific choice of parameters $\exfd$ and $\bxfd$.
Subsequently, we get
\begin{align}
\label{eq:keyfuncFDlow}
\rskfd\left(\rrfd\right) \overset{(b)}{\geq} \rskfd\left(\rrfdlow\right) =: \rskfdlow,
\end{align}
where $(b)$ follows from~\eqref{eq:rpFdBound}. 
The computation of the parameters $\exfd$, and $\bxfd$, is provided in Appendix~\ref{sec:appGenFD}. 
\vspace{.5em}
\begin{remark}
The particular structure of bounds in~\eqref{eq:keyfuncHDup} and~\eqref{eq:keyfuncFDlow} provides the worst-case performance gains between the FD and HD mode, i.e.,  $\rskfdlow-\rskhdup$. 
\end{remark}

\section{Main Results and Discussion}
\label{sec:results}
In this section, we present our main results by comparing the performance of HD and FD approaches.  We first define a metric which we refer to as improvement ratio, that shows the  performance gain of FD over HD approaches.
Typically, the upgrade from HD to FD capability requires some extra hardware at transceivers~\cite{src:heino2015recent}. Therefore, the improvement of FD mode over HD has to be quite significant in order to justify the additional expenses.

\begin{definition}
The FD over HD  improvement ratio $\eta$ and its lower bound $\underline{\eta}$ are given by 
\begin{align}
\label{eq:improvratio}
\eta = \frac{\rskfd-\rskhd}{\rskhd} \geq \underline{\eta} := \frac{\rskfdlow-\rskhdup}{\rskhdup},
\end{align}
where $\rskfd\neq 0$ and $\rskfdlow\geq \rskhdup$.
\end{definition}
\noindent Figure~\ref{fig:CompVsSnr} shows examples of secret-key rate and improvement ratio. We fix parameters to $\rho_{ae}=\rho_{be}=\rho_{e}=0.4$, $\rho_{ba}^2=1$, $\delta=0.97$ and $\beta=0.5$. In this example, we set $\snrae=\snrbe=\snr$ for the following reason --- If the legitimate nodes increase their transmit powers when they exchange pilots, then Eve benefits likewise in terms of $\snr$. Therefore, if Eve has finite $\snrae$ and $\snrbe$, it is reasonable to assume that both $\snrae$, and $\snrbe$ are proportional to $\snr$.  We show the FD rates~$\eqref{eq:keyfuncFDlow}$ for different values of $\alpha\in \{ -10, -17, -20\}$~dB. Fig.~\ref{fig:etaCompVsSnr} illustrates improvement ratio~\eqref{eq:improvratio} versus $\snr$ for  $\alpha\in \{ -17, -20\}$~dB. The Figure~\ref{fig:skCompVsSnr} shows the significance of self-interference cancellation, since we only achieve some improvement of FD over HD mode, if $\alpha\ll 1$. For $\alpha=-10$~dB, the HD mode performs better than the FD mode. However, perfect cancellation is \emph{not} necessary in order to justify the FD setup, and therefore secret-key generation in FD mode is attractive even if the methods for interference cancellation have low complexity. 
\ifCLASSOPTIONdraftcls
\else
\newcounter{MYtempeqncnt}
\begin{figure*}
\normalsize
\setcounter{equation}{38}
\begin{equation}
\label{eq:betaSecDer}
\frac{\partial^2R_{sk}}{\partial \beta^2}=
-\frac{
2\ln \tilde{r}_p4^{\tilde{r}_p\left(1+\frac{1}{\beta}\right)} (\lVert\bm b_x\rVert^2-\lVert\bm e_x\rVert^2)(1+\lVert\bm b_x\rVert^2)
}{
\beta^3\left[4^{\tilde{r}_p}(\lVert\bm b_x\rVert^2-\lVert\bm e_x\rVert^2)-4^{\frac{\tilde{r}_p}{\beta}}
(1+\lVert\bm b_x\rVert^2) 
\right]^2
}
<0
\end{equation}
\hrulefill
\vspace*{4pt}
\end{figure*} 
\begin{figure*}
\normalsize
\setcounter{equation}{39}
\begin{equation}
\label{eq:eSecDer}
\frac{\partial^2R_{sk}}{\partial (\lVert\bm e_x\rVert^2)^2}=
\frac{
(1+\lVert\bm b_x\rVert^2)\left(1-4^{-\frac{R_r}{\beta}}\right)
\left[1+4^{-\frac{R_r}{\beta}}(1+2\lVert\bm e_x\rVert^2)
+\lVert\bm b_x\rVert^2\left(1-4^{-\frac{R_r}{\beta}}\right)\right]
}{
(1+\lVert\bm e_x\rVert^2)^2
\left[ 1+4^{-\frac{R_r}{\beta}}\lVert\bm e_x\rVert^2
+\lVert\bm b_x\rVert^2\left(1-4^{-\frac{R_r}{\beta}}\right)
\right]^2
}>0
\end{equation}
\hrulefill
\vspace*{4pt}
\end{figure*} 
\fi

\ifCLASSOPTIONdraftcls
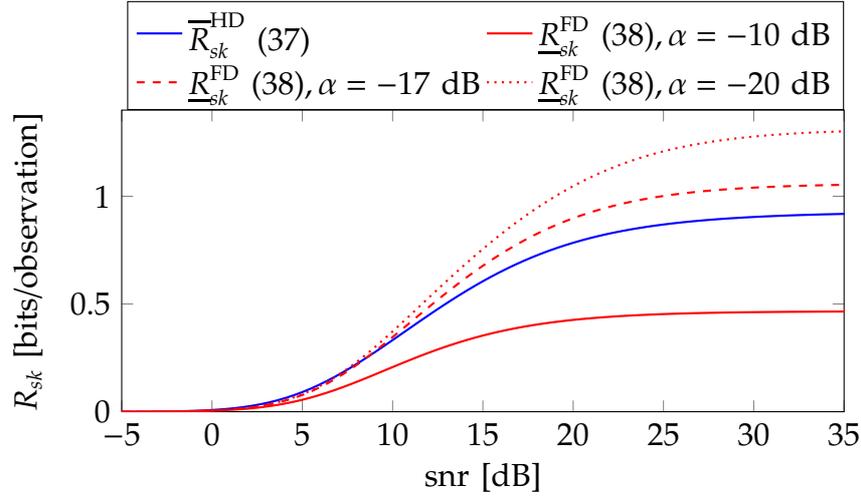
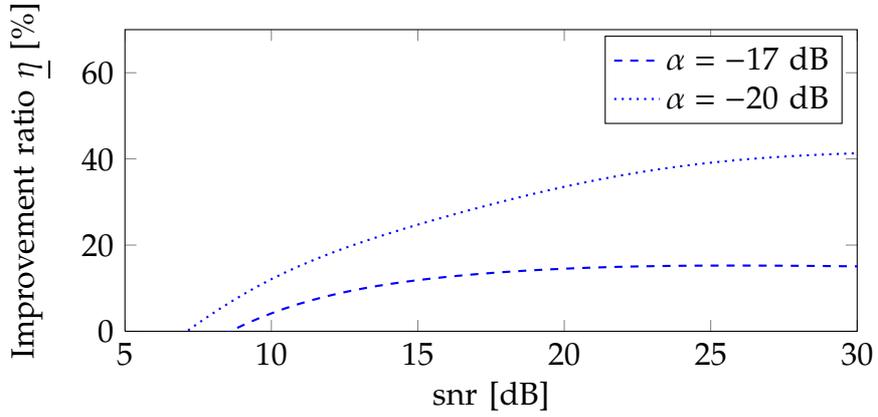
\begin{figure*}
\centering
  \subfloat[]{ \label{fig:skCompVsSnr}

\begin{tikzpicture}
\begin{axis}[%
width=9.5cm,
height=4cm,
scale only axis,
xmin=-5,
xmax=35,
xlabel={$\snr$ [dB]},
ymin=0,
ymax=1.4,
ylabel={$R_{sk}$ [bits/observation]},
legend columns=2,
legend style={draw=black,fill=white,legend cell align=left, at={(1,1)},	anchor=south east}
]
\addplot [color=blue,solid,thick]
  table{skCompVsSnr-1.tsv};
\addlegendentry{$\rskhdup~\eqref{eq:keyfuncHDup}$};

\addplot [color=red,solid,thick]
  table{skCompVsSnr-2.tsv};
\addlegendentry{$\rskfdlow~\eqref{eq:keyfuncFDlow}, \alpha=-10$ dB};
\addplot [color=red,solid,thick,dashed]
  table{skCompVsSnr-3.tsv};
\addlegendentry{$\rskfdlow~\eqref{eq:keyfuncFDlow}, \alpha=-17$ dB};
\addplot [color=red,solid,thick,dotted]
  table{skCompVsSnr-4.tsv};
\addlegendentry{$\rskfdlow~\eqref{eq:keyfuncFDlow}, \alpha=-20$ dB};

%
\end{axis}
\end{tikzpicture}%
 }
\hspace{1cm}
~\subfloat[]{\label{fig:etaCompVsSnr} 
\begin{tikzpicture}

\begin{axis}[%
width=9.63cm,
height=4cm,
scale only axis,
xmin=5,
xmax=30,
xlabel={$\snr$ [dB]},
ymin=0,
ymax=70,
ylabel={Improvement ratio $\underline{\eta}$ [\%]},
legend style={draw=black,fill=white,legend cell align=left}
]
\addplot [color=blue,solid,thick,dashed]
  table[row sep=crcr]{etaCompVsSnr-2.tsv};
\addlegendentry{$\alpha=-17$ dB};
\addplot [color=blue,solid,thick,dotted]
  table[row sep=crcr]{etaCompVsSnr-3.tsv};
\addlegendentry{$\alpha=-20$ dB};
\end{axis}
\end{tikzpicture}%
}
\caption{Secret-key rates (a) from bounded equivalents of HD~\eqref{eq:keyfuncHDup}, and FD modes~\eqref{eq:keyfuncFDlow}; and,  (b) improvement ratios~\eqref{eq:improvratio} for different values of $\snr$ at the legitimate users and $\snrae=\snrbe=\snr$ for the eavesdropper. The other parameters are set to $\rho_{ae}=\rho_{be}=\rho_{e}=0.4$, $\rho_{ba}^2=1$, $\delta=0.97$ and $\beta=0.5$.}
\label{fig:CompVsSnr}
\end{figure*}
\else
\begin{figure*}
\centering
  \subfloat[]{ \label{fig:skCompVsSnr}

\begin{tikzpicture}[spy using outlines={rectangle, magnification=2,connect spies}]
\begin{axis}[%
width=6cm,
height=4cm,
scale only axis,
xmin=-5,
xmax=35,
xlabel={$\snr$ [dB]},
ymin=0,
ymax=1.4,
ylabel={$R_{sk}$ [bits/observation]},
legend columns=2,
legend style={draw=black,fill=white,legend cell align=left, at={(1,1)},	anchor=south east}
]
\addplot [color=blue,solid,thick]
  table{../figures/skCompVsSnr-1.tsv};
\addlegendentry{$\rskhdup~\eqref{eq:keyfuncHDup}$};

\addplot [color=red,solid,thick]
  table{../figures/skCompVsSnr-2.tsv};
\addlegendentry{$\rskfdlow~\eqref{eq:keyfuncFDlow}, \alpha=-10$ dB};
\addplot [color=red,solid,thick,dashed]
  table{../figures/skCompVsSnr-3.tsv};
\addlegendentry{$\rskfdlow~\eqref{eq:keyfuncFDlow}, \alpha=-17$ dB};
\addplot [color=red,solid,thick,dotted]
  table{../figures/skCompVsSnr-4.tsv};
\addlegendentry{$\rskfdlow~\eqref{eq:keyfuncFDlow}, \alpha=-20$ dB};

%
\end{axis}
\end{tikzpicture}%
 }
\hspace{1cm}
~\subfloat[]{\label{fig:etaCompVsSnr} 
\begin{tikzpicture}

\begin{axis}[%
width=6cm,
height=4cm,
scale only axis,
xmin=5,
xmax=30,
xlabel={$\snr$ [dB]},
ymin=0,
ymax=70,
ylabel={Improvement ratio $\underline{\eta}$ [\%]},
legend style={draw=black,fill=white,legend cell align=left}
]
\addplot [color=blue,solid,thick,dashed]
  table[row sep=crcr]{../figures/etaCompVsSnr-2.tsv};
\addlegendentry{$\alpha=-17$ dB};
\addplot [color=blue,solid,thick,dotted]
  table[row sep=crcr]{../figures/etaCompVsSnr-3.tsv};
\addlegendentry{$\alpha=-20$ dB};
\end{axis}
\end{tikzpicture}%
}
\caption{Secret-key rates (a) from bounded equivalents of HD~\eqref{eq:keyfuncHDup}, and FD modes~\eqref{eq:keyfuncFDlow}; and,  (b) improvement ratios~\eqref{eq:improvratio} for different values of $\snr$ at the legitimate users and $\snrae=\snrbe=\snr$ for the eavesdropper. Parameters are set to $\rho_{ae}=\rho_{be}=\rho_{e}=0.4$, $\rho_{ba}^2=1$, $\delta=0.97$ and $\beta=0.5$.}
\label{fig:CompVsSnr}
\end{figure*}
\fi

\begin{remark}
Such nodes that are equipped with FD capability can always use the HD mode as a fallback solution. Therefore, a node in FD mode technically never performs worse than a node with HD mode only and effectively achieves $\max\left( \rskhd,\rskfd \right)$.   
\end{remark}

\subsection{Probing-reconciliation trade-off}
\label{sec:estComm}
Recall from~\figref{fig:seq} that the probing and  reconciliation phases are sharing $n$~coherence blocks, where the parameter $\beta$ captures the resource trade-off between the time duration of two phases. Next, we show that there exists an optimal $\beta$ for which the secret-key rate is maximized. From the system designer viewpoint, this value is especially relevant since --- except for the input power --- all other system parameters are difficult to change.  
\begin{property}
For any $0<\beta< 1$ and a reconciliation rate $R_{p}=(1-\beta)\tilde{r}_{p}$, where
\begin{enumerate}[label=(\roman*)]
\item $\tilde{r}_{p}$ is independent of $\beta$,
\item $\left\lVert\bm b_x \right\rVert^2>\left\lVert\bm e_x \right\rVert^2$,
\item $R_r>0$, 
\end{enumerate}
the key-reconciliation function $R_{sk}(R_r)$ is strictly concave with respect to $\beta$.
\end{property}
\begin{IEEEproof}
The key-reconciliation function of~\eqref{eq:keyfunc} has a extreme point within $0<\beta<1$ and its second derivative is given  by
\ifCLASSOPTIONdraftcls
\begin{equation*}
\label{eq:betaSecDer}
\frac{\partial^2R_{sk}}{\partial \beta^2}=
-\frac{
2\ln \tilde{r}_p4^{\tilde{r}_p\left(1+\frac{1}{\beta}\right)} (\lVert\bm b_x\rVert^2-\lVert\bm e_x\rVert^2)(1+\lVert\bm b_x\rVert^2)
}{
\beta^3\left[4^{\tilde{r}_p}(\lVert\bm b_x\rVert^2-\lVert\bm e_x\rVert^2)-4^{\frac{\tilde{r}_p}{\beta}}
(1+\lVert\bm b_x\rVert^2) 
\right]^2
}
<0.
\end{equation*}
\else 
~\eqref{eq:betaSecDer} on the next page.
\fi
The second derivative is negative for all $\beta$, which proves the concavity of~\eqref{eq:keyfunc} with respect to $\beta$.
\end{IEEEproof}
\vspace{.5em}
\noindent Figure \ref{fig:funcbeta} illustrates an example of the optimal  $\beta^{*}$ as a function of $\snr$. The optimal $\beta^{*}$ for HD and FD mode are obtained by maximizing $\rskhdup~\eqref{eq:keyfuncHDup}$ and $\rskfdlow~\eqref{eq:keyfuncFDlow}$ with respect to $\beta$.  It can be clearly seen from Figure~\ref{fig:funcbeta} that it is more beneficial to proportionally increase the number of channel estimates at higher $\snr$. In this regime, the channel supports enhanced reconciliation rates at higher $\snr$, and therefore the constraint of~\eqref{eq:commRate} is more relaxed. This trend can be observed for both HD and FD modes. The optimal $\beta^*$ is linearly increasing in HD mode at higher $\snr$. This implies that, there is no need of further investing into reconciliation, since the reconciliation rate gets unbounded with respect to $\snr$. In the case of FD mode, however, the reconciliation rate is impeded by self-interference and therefore limited for high $\snr$ and accordingly, $\beta^*$ approaches a saturation point.

\begin{remark}
 The proper choice of parameters like $\rho_{ba}$ or $\delta$ might not be  trivial, since they largely depend on the environment where the system is deployed; and, also on the properties of the node's transceivers. However, as many HD implementations (for instance~\cite{src:liu2013fast}) have shown, the channel reciprocity represented by $\delta\rho_{ba}$ is usually very high. Likewise, in all of our examples, we assume that $\delta$ is close to one. 
\end{remark}
\begin{figure*}
\centering
\begin{tikzpicture}
\begin{axis}[%
width=10cm,
height=3cm,
scale only axis,
xmin=0,
xmax=40,
xlabel={$\snr$ [dB]},
ymin=0.35,
ymax=0.9,
ylabel={$\beta^{*}$},
legend style={draw=black,fill=white,legend cell align=left, at={(0.8,0.1)}, anchor=south}
]
\addplot [thick,color=blue,thick]
  table{BetastarVsSnr-2.tsv};
\addlegendentry{HD}; 

\addplot [thick,color=red,thick]
  table{BetastarVsSnr-1.tsv}; 
\addlegendentry{FD};
\end{axis}
\end{tikzpicture}%
\caption{Depiction of optimal allocation $\beta^*$ of channel estimation and reconciliation resources for different values of the $\snr$ at the legitimate users where $\snrae=\snrbe=\snr$ for the eavesdropper. The other parameters are chosen as $\rho_{ba}=1$, $\delta=0.95$, $\rho_{ae}=\rho_{be}=\rho_e=0.4$ and $\alpha=-15~\text{dB}$.}
\label{fig:funcbeta}
\end{figure*}
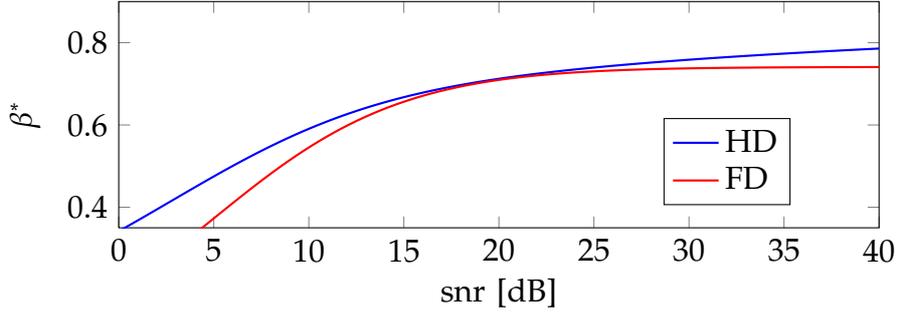

%
%

\subsection{Strong eavesdropper}
\label{sec:strongEve}
The strong secrecy condition~\eqref{eq:uniSecrecy} guarantees security against an Eve with unlimited computational power. However, in a practical system, Eve has more options. For instance, Eve can obtain equipment that provides almost noise-free reception of signals. Furthermore, even a passive Eve is able to change the position in the environment. Thus, Eve can evoke certain channel statistics that are more favorable for eavesdropping.  In this sub-section, we consider a special setting that models the worst-case scenario, where we assume that the receiver noise is negligible at Eve, i.e., $n_{ae}=n_{be}=n_{e}:=0$; and, provide regions under which a positive secret-key rate is achievable.  For convenience, we define the parameter $\xi:=\sqrt{ {\snrbe}/{\snrae}}$, that denotes the ratio of channel strengths between Bob-to-Eve and Alice-to-Eve links.

We first digress to provide a property which is useful to establish the results in this sub-section.
\begin{property}
\label{e-decresing}
The key-reconciliation function $R_{sk}(R_r)$~\eqref{eq:keyfunc} is strictly monotonically decreasing with respect to $\lVert\bm e_x\rVert^2$.
\end{property}
\begin{IEEEproof}
The first derivative of $R_{sk}(R_r)$~\eqref{eq:keyfunc} with respect to $\lVert\bm e_x\rVert^2$ does not yield extreme points as necessary condition of local extrema. After some straightforward algebra, by taking the second derivative 
\ifCLASSOPTIONdraftcls
\begin{equation}
 \label{eq:eSecDer}
\frac{\partial^2R_{sk}}{\partial (\lVert\bm e_x\rVert^2)^2}=
\frac{
(1+\lVert\bm b_x\rVert^2)\left(1-4^{-\frac{R_r}{\beta}}\right)
\left[1+4^{-\frac{R_r}{\beta}}(1+2\lVert\bm e_x\rVert^2)
+\lVert\bm b_x\rVert^2\left(1-4^{-\frac{R_r}{\beta}}\right)\right]
}{
(1+\lVert\bm e_x\rVert^2)^2
\left[ 1+4^{-\frac{R_r}{\beta}}\lVert\bm e_x\rVert^2
+\lVert\bm b_x\rVert^2\left(1-4^{-\frac{R_r}{\beta}}\right)
\right]^2
}>0,
\end{equation}
\else
as provided in~\eqref{eq:eSecDer}, 
\fi
it can be readily seen that~\eqref{eq:eSecDer} is positive for all $\lVert\bm e_x\rVert^2$, which yields the desired result.
\end{IEEEproof}
\vspace{.5em}
Next, we provide the regions under which the legitimate nodes experience the \emph{best} or the \emph{worst} situations. 

\subsubsection{HD mode}
In this section, we first compute $\exhd$ and later on provide the region under which positive secrecy rate is achievable. The computation of $\exhd$ follows along similar lines as shown in Appendix~\ref{sec:appGenHD} and is given by 
\ifCLASSOPTIONdraftcls
 \begin{align}
\label{eq:exStrongEveHd} 
 \exhd=\frac{\snr(\rho_{ae}^2+\rho_{be}^2-2\rho_{e}\rho_{ae}\rho_{be})}{(1+\snr)(1-\rho_e^2)-\snr(\rho_{ae}^2+\rho_{be}^2-2\rho_{e}\rho_{ae}\rho_{be})}.
 \end{align}
\else
 \begin{align}
 &\exhd =  \nonumber\\
 \label{eq:exStrongEveHd}
 &\quad\frac{\snr(\rho_{ae}^2+\rho_{be}^2-2\rho_{e}\rho_{ae}\rho_{be})}{(1+\snr)(1-\rho_e^2)-\snr(\rho_{ae}^2+\rho_{be}^2-2\rho_{e}\rho_{ae}\rho_{be})}.
 \end{align}
\fi
According to~\eqref{eq:azaz}, the covariance matrix of Eve given the observations at Alice is  
\begin{align}
 \label{eq:eveStrongVarsHD}
 \bm\Sigma_{\bm z|x} :=
  \begin{pmatrix}
   \Sigma_{11} & \Sigma_{12}  \\
   \Sigma_{21} &  \Sigma_{22} \\
  \end{pmatrix}
 \end{align}
 with  entries
 \begin{align*}
 \Sigma_{11}&=\snrae\left(1-\frac{\snr}{1+\snr}\rho^2_{ae}\right),\\
  \Sigma_{12}&=\Sigma_{21}=\sqrt{\snrae\snrbe}\left(\rho_{e}-\frac{\snr}{1+\snr}\rho_{ae}\rho_{be}\right),\\
   \Sigma_{22}&=\snrbe\left(1-\frac{\snr}{1+\snr}\rho^2_{be}\right),
 \end{align*}
 which needs to satisfy the  positive definite matrix condition.\footnote{A matrix is positive definite, if all leading principal minors are positive. Let $\bm A$ be a $n\times n$ matrix. A $k\times k$ sub matrix is constructed by deleting the last $n-k$ columns and rows from $\bm A$. The leading principal minor of order~$k$ is the determinant of that submatrix.} 

\begin{definition}
For a fixed $\snr$, the region $\mathcal{A}^{\text{HD}}_z$ denotes all 4-tuples $\left(\rho_{ae}, \rho_{be}, \rho_e,\xi\right)\in[-1,1]\times[-1,1]\times[-1,1]\times(0,\infty)$ that satisfy\footnote{Note the parameter $\xi=\sqrt{ {\snrbe}/{\snrae}}$ is added in order to enable the intersection to the corresponding region in the FD mode.}
\label{eq:feasibleHD}
\begin{align}
\label{eq:feasibleHD2}
1-\rho_e^2-\frac{\snr}{1+\snr}\left(\rho_{ae}^2+\rho_{be}^2-2\rho_{ae}\rho_{be}\rho_e\right)&>0.
\end{align} 
\end{definition}
\begin{proof}
The region $\mathcal{A}^{\text{HD}}_z$ in Definition~\ref{eq:feasibleHD} follows by computing the determinant of~\eqref{eq:eveStrongVarsHD}.
\end{proof}

\subsubsection{FD mode}
In this section, we first compute the $\exfd$ and later on provide the region under which positive secret-key rate is achievable. The computation of $\exfd$  is given by 
\ifCLASSOPTIONdraftcls
 \begin{align}
 \label{eq:exStrongEveFd}
 \exfd =\frac{\snr(\rho_{ae}+\xi\rho_{be})^2}{(1+\xi^2+2\xi\rho_e)(1+\snr(1+\alpha^2))-\snr(\rho_{ae}+\xi\rho_{be})^2}.
 \end{align}
\else
 \begin{align}
 \label{eq:exStrongEveFd}
 &\exfd = \notag\\ &\qquad\frac{\snr(\rho_{ae}+\xi\rho_{be})^2}{(1+\xi^2+2\xi\rho_e)(1+\snr(1+\alpha^2))-\snr(\rho_{ae}+\xi\rho_{be})^2}.
 \end{align}
\fi
The covariance matrix~\eqref{eq:azaz} yields
\begin{align}
  \label{eq:eveStrongVarsFD}
  \bm\Sigma_{\bm z|x}  &=\snr\left( 1+\xi^2+2\xi\rho_e\right)-\frac{\snr(\rho_{ae}+\xi\rho_{be})^2}{1+\snr(1+\alpha^2)}.
\end{align}

\begin{definition}
\label{eq:feasibleFD}
For fixed $\snr$ and $\alpha$, the region $\mathcal{A}^{\text{FD}}_z$ denotes all 4-tuples $\left(\rho_{ae}, \rho_{be}, \rho_e,\xi\right) \in[-1,1]\times[-1,1]\times[-1,1]\times(0,\infty)$ that satisfy
\begin{align}
1+\xi^2+2\xi\rho_e-\frac{\snr}{1+\snr(1+\alpha^2)}\left(\rho_{ae}+\xi\rho_{be}\right)^2>0.
\end{align} 
\end{definition}

\begin{proof}
The region defined in Definition~\ref{eq:feasibleFD} follows straightforwardly from~\eqref{eq:eveStrongVarsFD}.
\end{proof}
\vspace{.5em}

\noindent From the Definitions~\ref{eq:feasibleHD} and~\ref{eq:feasibleFD} --- the \emph{feasible} region   $\mathcal{A}_z$ --- is then given by  $\mathcal{A}_z=\mathcal{A}^{\text{HD}}_z\cap\mathcal{A}^{\text{FD}}_z$. Figure~\ref{fig:feasible} shows an example of the typical feasible region where the meshes denote the outer boundaries of the region.
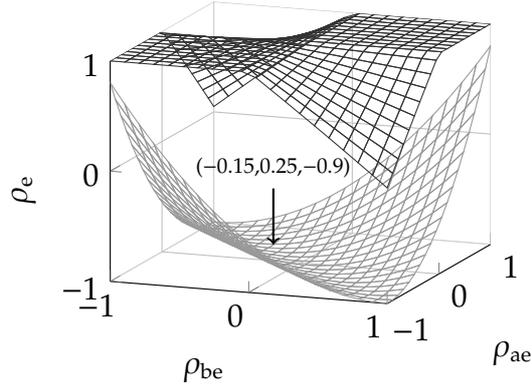
\begin{figure}
\centering
\begin{tikzpicture}
\begin{axis}[%
width=5cm,
height=4cm,
view={20.5}{16},
scale only axis,
xmin=-1,
xmax=1,
xlabel={$\rho{}_{\text{be}}$},
xmajorgrids,
ymin=-1,
ymax=1,
ylabel={$\rho{}_{\text{ae}}$},
ymajorgrids,
zmin=-1,
zmax=1,
zlabel={$\rho{}_{\text{e}}$},
zmajorgrids,
axis x line*=bottom,
axis y line*=left,
axis z line*=left
]

\addplot3[%
mesh,
draw=black!40,
mesh/rows=20]
table[header=false] {feasibleSetEve-1.tsv};

\addplot3[%
mesh,
draw=black!80,
mesh/rows=20]
table[header=false] {feasibleSetEve-2.tsv};

\node at (axis cs:-0.2525,0.1515,-0.2) (lmax) {$\scriptstyle\left(-0.15,0.25,-0.9\right)$};
\draw[thick, ->] (lmax) -- (axis
cs:-0.2525,0.1515,-0.9);

%
\end{axis}
\end{tikzpicture}%
\caption{Example of feasible region $\mathcal{A}_z$ for $\xi=1$ and fixed parameters $\alpha=-20$~dB, $\snr=10$~dB. The arrow points to a specific 3-tuple of $\left(\rho_{ae},\rho_{be},\rho_e\right)$ that is referred to later in Figure~\ref{fig:rateDiffVsRhoaebePos}.}
\label{fig:feasible}
\end{figure}
We now discuss two situations  where the legitimate nodes experience best or worst secret-key rates. The best case  occurs if Eve has  minimal knowledge of the channel state of Bob-to-Alice link. Recall from Property~\ref{e-decresing} that the secrecy rate $R_{sk}(R_r)$ in~\eqref{eq:keyfunc} is maximized,  by minimizing the $\lVert\bm e_x\rVert^2$.  The  best case  in terms of secret-key rate trivially happens with $\rho_{ae}=\rho_{be}=0$ in both modes. In this case, since there is no correlation between the Alice-to-Bob and  Alice-to-Eve links and vice versa; thus, Eve is unable to get any information about the secret key. In addition to this, it is interesting to note that there are some other useful situations in which positive secret-key rates are achievable. In HD mode, a maximum of secret-key rate is attained if $\rho_{ae}=\rho_e\rho_{be}$ holds with fixed $\rho_e$, and   $\rho_{be}$, or symmetrically, if $\rho_{be}=\rho_e\rho_{ae}$ holds with fixed $\rho_e$, and  $\rho_{ae}$. In FD mode, $\exfd$ vanishes if $\rho_{ae}=-\xi\rho_{be}$ holds. This follows due to the fact that Eve only obtains the superposition of observations from Alice-to-Eve and Bob-to Eve links, which can nullify the overall observation at Eve. 

In contrast to the above, the worst case  happens if Eve has occupied the most favorable condition for eavesdropping by maximizing $\lVert\bm e_x\rVert^2$ in~\eqref{eq:keyfunc}, which in turn reduces the secret key rate.  In HD mode, it can be readily shown that $\exhd$ is convex with respect to the tuple $(\rho_{ae}, \rho_{be})$, so the maximum  lies at the boundary of the feasible region defined in Definition~\ref{eq:feasibleHD}. In the case of FD mode, one can deduce a similar conclusion from~\eqref{eq:exStrongEveFd} and Definition~\ref{eq:feasibleFD}.

\begin{figure*}
\centering
\subfloat[$\rho_e=-0.9$]{ 
\label{fig:rateDiffVsRhoaebeNeg}
%
%
%
\definecolor{mycolor1}{rgb}{0.00000,0.00000,0.56250}%
\definecolor{mycolor2}{rgb}{0.00000,0.00000,0.81250}%
\definecolor{mycolor3}{rgb}{0.00000,0.43750,1.00000}%
\definecolor{mycolor4}{rgb}{0.06250,1.00000,0.93750}%
\definecolor{mycolor5}{rgb}{0.68750,1.00000,0.31250}%
\definecolor{mycolor6}{rgb}{1.00000,0.68750,0.00000}%
\definecolor{mycolor7}{rgb}{1.00000,0.06250,0.00000}%
\begin{tikzpicture}

\begin{axis}[%
width=8cm,
height=4cm,
scale only axis,
xmin=-1,
xmax=1,
xlabel={$\rho_{be}$},
ymin=-1,
ymax=1,
ylabel={$\rho_{ae}$},
colormap/jet,
colorbar,
colorbar style={
	title={$\rskfdlow-\rskhdup$}
},
point meta min=0.0246538887719323,
point meta max=0.35
]

\addplot[area legend,solid,fill=mycolor1,draw=black,forget plot]
table[row sep=crcr] {%
rateDiffVsRhoaebeNeg-1.tsv};

\addplot[area legend,solid,fill=mycolor2,draw=black,forget plot]
table[row sep=crcr] {%
rateDiffVsRhoaebeNeg-2.tsv};

\addplot[area legend,solid,fill=mycolor2,draw=black,forget plot]
table[row sep=crcr] {%
rateDiffVsRhoaebeNeg-3.tsv};

\addplot[area legend,solid,fill=mycolor3,draw=black,forget plot]
table[row sep=crcr] {%
rateDiffVsRhoaebeNeg-4.tsv};

\addplot[area legend,solid,fill=mycolor3,draw=black,forget plot]
table[row sep=crcr] {%
rateDiffVsRhoaebeNeg-5.tsv};

\addplot[area legend,solid,fill=mycolor4,draw=black,forget plot]
table[row sep=crcr] {%
rateDiffVsRhoaebeNeg-6.tsv};

\addplot[area legend,solid,fill=mycolor4,draw=black,forget plot]
table[row sep=crcr] {%
rateDiffVsRhoaebeNeg-7.tsv};

\addplot[area legend,solid,fill=mycolor5,draw=black,forget plot]
table[row sep=crcr] {%
rateDiffVsRhoaebeNeg-8.tsv};

\addplot[area legend,solid,fill=mycolor5,draw=black,forget plot]
table[row sep=crcr] {%
rateDiffVsRhoaebeNeg-9.tsv};

\addplot[area legend,solid,fill=mycolor6,draw=black,forget plot]
table[row sep=crcr] {%
rateDiffVsRhoaebeNeg-10.tsv};

\addplot[area legend,solid,fill=mycolor6,draw=black,forget plot]
table[row sep=crcr] {%
rateDiffVsRhoaebeNeg-11.tsv};

\addplot[area legend,solid,fill=mycolor7,draw=black,forget plot]
table[row sep=crcr] {%
rateDiffVsRhoaebeNeg-12.tsv};

\addplot[area legend,solid,fill=mycolor7,draw=black,forget plot]
table[row sep=crcr] {%
rateDiffVsRhoaebeNeg-13.tsv};

\addplot[area legend,solid,fill=black!50!red,draw=black,forget plot]
table[row sep=crcr] {%
rateDiffVsRhoaebeNeg-14.tsv};

\addplot[area legend,solid,fill=black!50!red,draw=black,forget plot]
table[row sep=crcr] {%
rateDiffVsRhoaebeNeg-15.tsv};

\addplot[area legend,solid,fill=mycolor2,draw=black,forget plot]
table[row sep=crcr] {%
rateDiffVsRhoaebeNeg-16.tsv};

\addplot[area legend,solid,fill=mycolor2,draw=black,forget plot]
table[row sep=crcr] {%
rateDiffVsRhoaebeNeg-17.tsv};

\addplot[area legend,solid,fill=mycolor2,draw=black,forget plot]
table[row sep=crcr] {%
rateDiffVsRhoaebeNeg-18.tsv};

\addplot[area legend,solid,fill=mycolor2,draw=black,forget plot]
table[row sep=crcr] {%
rateDiffVsRhoaebeNeg-19.tsv};

\addplot[area legend,solid,fill=mycolor3,draw=black,forget plot]
table[row sep=crcr] {%
rateDiffVsRhoaebeNeg-20.tsv};

\addplot[area legend,solid,fill=mycolor3,draw=black,forget plot]
table[row sep=crcr] {%
rateDiffVsRhoaebeNeg-21.tsv};

\addplot[area legend,solid,fill=mycolor3,draw=black,forget plot]
table[row sep=crcr] {%
rateDiffVsRhoaebeNeg-22.tsv};

\addplot[area legend,solid,fill=mycolor3,draw=black,forget plot]
table[row sep=crcr] {%
rateDiffVsRhoaebeNeg-23.tsv};

\addplot[area legend,solid,fill=mycolor6,draw=black,forget plot]
table[row sep=crcr] {%
rateDiffVsRhoaebeNeg-24.tsv};

\addplot[area legend,solid,fill=mycolor6,draw=black,forget plot]
table[row sep=crcr] {%
rateDiffVsRhoaebeNeg-25.tsv};

\addplot[area legend,solid,fill=mycolor6,draw=black,forget plot]
table[row sep=crcr] {%
rateDiffVsRhoaebeNeg-26.tsv};

\addplot[area legend,solid,fill=mycolor6,draw=black,forget plot]
table[row sep=crcr] {%
rateDiffVsRhoaebeNeg-27.tsv};

\end{axis}
\end{tikzpicture}%
}
\\
\subfloat[$\rho_e=0.9$]{
\label{fig:rateDiffVsRhoaebePos} 
\input{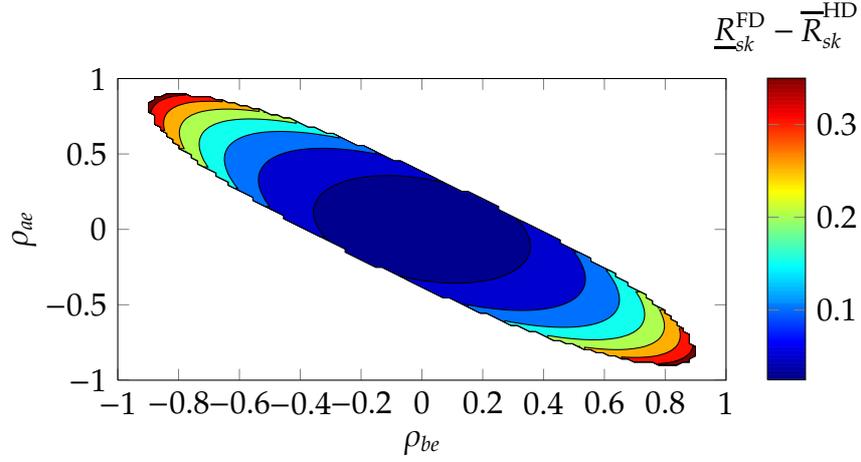}
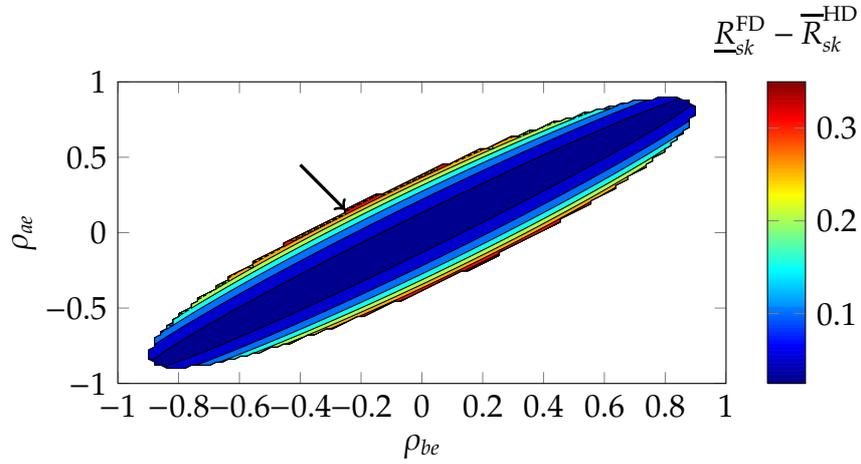
}
\caption{Improvement of FD over HD mode for parameters $\xi=1$, $\alpha=-15$~dB, $\snr=10$~dB, $\rho_{ba}^2=\delta=1$ and $\beta=0.5$, measured in bits/observation. The arrow in (b) denotes the point where the maximum value is achieved. The corresponding 3-tuple of  $\left(\rho_{ae},\rho_{be},\rho_e\right)$ is marked by the arrow in  Figure~\ref{fig:feasible}.}
\label{fig:VsRhoae}
\end{figure*}
\figref{fig:VsRhoae} shows the impact of different choices of $\rho_e, \rho_{ae}$, and $\rho_{be}$ on the $\rskfdlow-\rskhdup$ term, i.e., the improvement of FD over HD mode. Recall that the parameter $\rho_e$ denotes the mutual correlation of channels Alice-to-Eve and Bob-to-Eve. It also comprises statistical correlation that is not related to the Alice-to-Bob link. As it can be seen from \figref{fig:rateDiffVsRhoaebeNeg}, the FD system offers the most advantage if the correlation coefficients $\rho_{ae}$ and $\rho_{be}$ have different signs. Due to the superposition of observations from Alice-to-Eve and Bob-to-Eve links, for $\rho_e<0$ they are more likely to cancel each other. If we have $\rho_e>0$ as depicted in \figref{fig:rateDiffVsRhoaebePos}, the observations are likely to add constructively, and therefore significant improvement of FD over HD mode is only apparent if the values of $\rho_{ae}$ and $\rho_{be}$ are close to zero. \figref{fig:rateDiffVsRhoaeXi} shows the impact of $\xi$ and $\rho_{ae}$ on the difference of secret-key rates. For numerical tractability, we show it with fixed $\rho_{be}$ and $\rho_e$. 
\begin{figure}
\centering
\begin{tikzpicture}
\begin{axis}[%
width=8cm,
height=4cm,
unbounded coords=jump,
view={-37.5}{30},
scale only axis,
xmin=-20,
xmax=20,
xlabel={$\xi$~[dB]},
xmajorgrids,
ymin=-1,
ymax=1,
ylabel={$\rho_{ae}$},
ymajorgrids,
zmin=0,
zmax=0.4,
zlabel={$\rskfdlow-\rskhdup$ [bits/observation]},
zmajorgrids,
axis x line*=bottom,
axis y line*=left,
axis z line*=left
]

\addplot3[%
surf,
shader=faceted,
draw=black,
colormap/jet,
mesh/rows=25]
table[row sep=crcr,header=false] {rateDiffVsRhoaeXi-1.tsv};
\end{axis}
\end{tikzpicture}
\caption{Improvement of FD over HD mode for parameters $\rho_e=0.2$, $\rho_{be}=-0.4$, $\alpha=-15$~dB, $\snr=10$~dB, $\rho_{ba}^2=\delta=1$ and $\beta=0.5$.}
\label{fig:rateDiffVsRhoaeXi}
\end{figure}
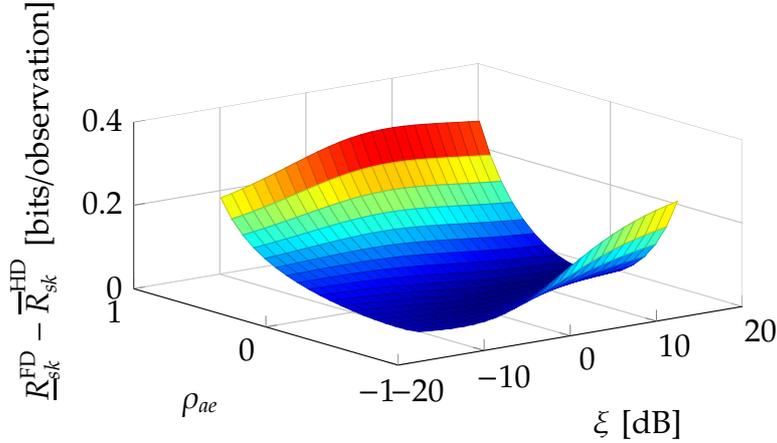
 The maximum gain is obtained when $\xi\approx 5$~dB. As mentioned before, $\rho_{ae}=-\xi\rho_{be}$ minimizes $\exfd$ in FD mode. Again, the FD mode is most beneficial if the correlation coefficients $\rho_{ae}$ and $\rho_{be}$ have opposite signs. 

\subsection{High SNR}
\label{sec:highSnr}
In this sub-section, we examine the case for arbitrarily high SNR at both the legitimate nodes and at the eavesdropper, where we ignore the contribution of noise at each node. For simplicity of analysis, we assume that
\begin{enumerate}
\item The eavesdropper experiences symmetric observations, i.e.,  $\rho_{ae}=\rho_{be}$ 
\item The users are able to perfectly cancel any self-interference, i.e., $\alpha=0$.
\end{enumerate}
 
\noindent We now provide a proposition that shows that under which conditions the FD mode performs better than HD mode.
\begin{proposition}
\label{sec:propsuffhighsnr}
The improvement ratio $\underline{\eta}>0$ of the FD over HD mode in the high SNR regime holds, if the following relations are satisfied   
\begin{align}
\label{eq:suffCond1}
\delta^2\rho_{ba}^2 &> \frac{2\rho_{ae}^2}{1+\rho_e-2\rho_{ae}^2}, \\
\label{eq:suffCond2}
\rho_{e} &< 1.
\end{align}
\end{proposition}
\begin{IEEEproof}
We begin the proof as follows. In order to obtain a positive secret-key rate, we need to fulfill only Property~\ref{sec:propskpos}-\ref{property-sec}, i.e.,  $\bxhd>\exhd$, since  Property~\ref{sec:propskpos}-\ref{property-R} always holds because of insignificant contribution of noise. By appropriately computing the parameters $\bxhd$ and $\exhd$ along similar lines as shown in~\eqref{b-hd} and~\eqref{eq:exhdfull} from Appendix~\ref{sec:appGeneral}, we get
\begin{align}
\label{eq:cond1ineq}
\frac{\delta^2\rho_{ba}^2}{1-\delta^2\rho_{ba}^2}>
\frac{2\rho_{ae}^2}{1+\rho_e-2\rho_{ae}^2}.
\end{align}
The term in~\eqref{eq:suffCond1} then follows by lower bounding the LHS of~\eqref{eq:cond1ineq}.

We next bound the inequality \eqref{eq:suffCond2} as follows. Recall that the improvement requirement  $\underline{\eta}>0$ implies that $\rskfdlow>\rskhdup$. Starting  from~\eqref{eq:keyfuncHDup} and~\eqref{eq:keyfuncFDlow}, we get
\begin{align}
\label{eq:cond2ineq}
\frac{1+\bxhd}{1+\exhd}>\frac{1+\bxfd}{1+\exfd}.
\end{align}
 Subsequently, by plugging~\eqref{b-hd},~\eqref{eq:exhdfull},~\eqref{b-fd} and~\eqref{eq:exfdfull} from Appendix~\ref{sec:appGeneral} into~\eqref{eq:cond2ineq}, we obtain
 \ifCLASSOPTIONdraftcls
 \begin{align}
 \frac{1+\rho_e}{(1-\rho^2_{ba})(1+\rho_e-2\rho^2_{ae})}> 
 \label{eq:suffCondProof1}
\frac{(1+\xi^2+2\rho_e\xi)}{(1-\delta^2\rho^2_{ba})\left[ (1+\xi^2+2\rho_e\xi)-\rho^2_{ae}(1+\xi)^2 \right]}.
 \end{align}
 \else
\begin{align}
&\frac{1+\rho_e}{(1-\rho^2_{ba})(1+\rho_e-2\rho^2_{ae})}> \notag\\
\label{eq:suffCondProof1}
&\qquad\frac{(1+\xi^2+2\rho_e\xi)}{(1-\delta^2\rho^2_{ba})\left[ (1+\xi^2+2\rho_e\xi)-\rho^2_{ae}(1+\xi)^2 \right]}.
\end{align}
 \fi
Then~\eqref{eq:suffCondProof1} can be written as 
\ifCLASSOPTIONdraftcls
\begin{align}
\label{eq:suffCondProof2}
\frac{(1-\delta^2\rho^2_{ba})}{(1-\rho^2_{ba})}\frac{(1+\rho_e)}{(1+\rho_e-2\rho^2_{ae})}> 
\frac{(1+\xi^2+2\rho_e\xi)}{\left[ (1+\xi^2+2\rho_e\xi)-\rho^2_{ae}(1+\xi)^2 \right]}.
\end{align}
\else
\begin{align}
&\frac{(1-\delta^2\rho^2_{ba})}{(1-\rho^2_{ba})}\frac{(1+\rho_e)}{(1+\rho_e-2\rho^2_{ae})}> \notag\\
\label{eq:suffCondProof2}
&\qquad\frac{(1+\xi^2+2\rho_e\xi)}{\left[ (1+\xi^2+2\rho_e\xi)-\rho^2_{ae}(1+\xi)^2 \right]}.
\end{align}
\fi

\noindent Next, by lower bound the term on the LHS in~\eqref{eq:suffCondProof2} we get
\begin{align}
\label{eq:suffCondProof3}
\frac{1+\rho_e}{1+\rho_e-2\rho^2_{ae}}> \frac{(1+\xi^2+2\rho_e\xi)}{ (1+\xi^2+2\rho_e\xi)-\rho^2_{ae}(1+\xi)^2}.
\end{align}
Finally, by rearranging~\eqref{eq:suffCondProof3} with respect to $\rho_e$, we get~\eqref{eq:suffCond2}. Note that for the special cases $\xi=1$ and $\rho_{ae}=0$, $\rskfdlow>\rskhdup$ is always fulfilled. 
\end{IEEEproof}
\begin{remark}
The condition~\eqref{eq:suffCond1} ensures that some advantage of the legitimate users over the eavesdropper exists. Without self-interference and delay penalty, i.e., $\alpha=0$ and $\delta=1$, FD almost always performs better than HD mode, which has already been observed in the literature~\cite{src:khisti2012interactive}.
\end{remark}
\section{Summary}
\label{sec:conclusion}
In this work, we studied a secret-key generation setup for reciprocal wireless channels for nodes with half-duplex~(HD) and full-duplex~(FD) capabilities. We first developed a system model that captured the channel probing as well as the public communication overhead required for key reconciliation. Next, we  established a key-reconciliation function that is used as a metric for comparison between HD and FD modes. The analysis revealed the improvements in secret-key rate in FD mode over the HD mode even under the impact of self-interference. We analyzed the probing-reconciliation trade-off, which holds an optimal solution in both modes. In the case of a strong eavesdropper, we identified certain situations of channel statistics that are either most favorable or detrimental for the legitimate users. At the high SNR regime, we established a sufficient condition that guarantees superior performance of the FD over the HD mode.  From a system designer viewpoint, the results provide insight under which the FD can give certain gains over the conventional HD mode. 

\appendices
\section{Key Computation Parameters in~\eqref{eq:keyfuncHDup} and~\eqref{eq:keyfuncFDlow}}
\label{sec:appGeneral}
We now provide the computation of parameters which are required to establish the key computations functions for HD and FD modes in sub-sections~\ref{HD-Key-compuation} and~\ref{FD-Key-compuation}, respectively. Let $\sigma_x^2 := \expv\left[ x^2\right]$, $\sigma_y^2 := \expv\left[ y^2\right]$, $\sigma_{yx} := \expv\left[ xy\right]$ and $\sigma^2_{y|x}:=\expv\left[ y^2|x\right]$. The results are obtained under the condition that $\bm\Sigma_{\bm z}^{-1}$, $\bm\Sigma_{\bm z|x}^{-1}$ and $\sigma_{y|x}^{-2}$ exist in both modes.
\subsubsection{HD mode}
\label{sec:appGenHD}
In HD mode the parameters $\sigma_x^2$, $\sigma_y^2$ and $\sigma_{yx} $ are given by
\begin{align*}
\sigma_x^2 &= \sigma_y^2 = 1+\snr, \\
\sigma_{yx} &= \snr\delta\rho_{ba}.
\end{align*}
Next we compute the parameters $\bxhd$ and $\exhd$ as follows. Starting from~\eqref{eq:parameterb} and~\eqref{pb}, we get
\begin{align}
\label{b-hd}
\bxhd &=
\frac{\sigma_{yx}^2\sigma_{y}^{-2}
}{\sigma^2_x-\sigma_{yx}\sigma_{y}^{-2}}\nonumber\\
&=\frac{\snr^2\delta^2\rho_{ba}^2}{(1+\snr)^2-\snr^2\delta^2\rho_{ba}^2}. 
\end{align}
Similarly, we can compute the parameters $\exhd$ as follows. From~\eqref{eq:parametere} and~\eqref{pe}, we get
\begin{align}
\label{mod-e-sq}
\exhd &=
\frac{\bm\sigma_{\bm zx}^T\bm\Sigma_{\bm z}^{-1}\bm\sigma_{\bm zx}
}{\sigma^2_x-\bm\sigma_{\bm zx}^T\bm\Sigma_{\bm z}^{-1}\bm\sigma_{\bm zx}},
\end{align}
where
\begin{align}
\label{sigma-z}
\bm\Sigma_{\bm z} &=
 \begin{pmatrix}
  1+\snrae & \sqrt{\snrae\snrbe}\rho_{e}  \\
  \sqrt{\snrae\snrbe}\rho_{e} &  1+\snrbe \\
 \end{pmatrix}, \\
 \label{sigma-zx}
\bm\sigma_{\bm zx}&=\sqrt{\snr}
 \begin{pmatrix}
  \sqrt{\snrae}\rho_{ae}   \\
  \sqrt{\snrbe}\rho_{be} \\
 \end{pmatrix}.
 \end{align}
 Then plugging in the values of~\eqref{sigma-z} and~\eqref{sigma-zx} in~\eqref{mod-e-sq} we get
 \ifCLASSOPTIONdraftcls
	\begin{align}
	\label{eq:exhdfull}
	\exhd = \frac{\text{N}_{\text{HD}}}{(1+\snr)\left[(1+\snrae)(1+\snrbe)-\snrae\snrbe\rho_e^2\right]-\text{N}_{\text{HD}}},
	\end{align}
 \else
~\eqref{eq:exhdfull} shown on top of the next page,
\begin{figure*}	
	\setcounter{equation}{57}
	\begin{align}
	\label{eq:exhdfull}
	\exhd = \frac{\text{N}_{\text{HD}}}{(1+\snr)\left[(1+\snrae)(1+\snrbe)-\snrae\snrbe\rho_e^2\right]-\text{N}_{\text{HD}}}
	\end{align}
	\hrulefill
\end{figure*}
 \fi

\noindent where $\text{N}_{\text{HD}}:=\snr(1+\snrbe)\snrae\rho_{ae}^2
+\snr(1+\snrae)\snrbe\rho_{be}^2
-2\snrae\snrbe\snr\rho_{ae}\rho_{be}\rho_{e}$.

\subsubsection{FD mode}
\label{sec:appGenFD}
In FD mode the parameters $\sigma_x^2$, $\sigma_y^2$ and $\sigma_{yx} $ are given by
\begin{align*}
\sigma_x^2 &= \sigma_y^2 = 1+\snr(1+\alpha^2), \\
\sigma_{yx} &= \snr\delta\rho_{ba}.
\end{align*}
Next, we compute the parameters $\bxfd$ and $\exfd$ as follows. Starting from~\eqref{eq:parameterb} and~\eqref{pb}, we get
\begin{align}
\label{b-fd}
\bxfd &=
\frac{\snr^2\rho_{ba}^2}{(1+\snr(1+\alpha^2))^2-\snr^2\rho_{ba}^2}. 
\end{align}
Similarly, we can compute the parameters $\exhd$ as follows. From~\eqref{eq:parametere} and~\eqref{pe}, we get 
 \ifCLASSOPTIONdraftcls
 	\begin{align}
 	\label{eq:exfdfull}
 	\exfd=\frac{\text{N}_{\text{FD}}}{\left[1+\snr(1+\alpha^2)\right](1+\snrae+\snrbe+2\sqrt{\snrae\snrbe}\rho_e)-\text{N}_{\text{FD}}},
 	\end{align}
 \else
 ~\eqref{eq:exfdfull} shown on top of the next page,
 \begin{figure*}	
 	\setcounter{equation}{59}
 	\begin{align}
 	\label{eq:exfdfull}
 	\exfd=\frac{\text{N}_{\text{FD}}}{\left[1+\snr(1+\alpha^2)\right](1+\snrae+\snrbe+2\sqrt{\snrae\snrbe}\rho_e)-\text{N}_{\text{FD}}}
 	\end{align}
 	\hrulefill
 \end{figure*}
 \fi

\noindent where $\text{N}_{\text{FD}}=\snr(\sqrt{\snrae}\rho_{ae} + \sqrt{\snrbe}\rho_{be})^2$, $\bm\Sigma_{\bm z}=1+\snrae+\snrbe+2\sqrt{\snrae\snrbe}\rho_e$ and  $\bm\sigma_{zx}=\sqrt{\snr\snrae}\rho_{ae} + \sqrt{\snr\snrbe}\rho_{be}$.

\balance

\bibliographystyle{IEEEtran}
\bibliography{IEEEabrv,Conf_abrv_new,references}

\listoftodos

\end{document}